\documentclass[10pt,a4paper]{scrartcl}
\usepackage{graphicx}

\usepackage{amsmath, amsthm, amssymb, tikz, xspace, graphicx, array, rotating, wasysym, multirow, todonotes,enumitem,booktabs,url}
\usepackage[round]{natbib}

\author{Marie-Louise Lackner\footnote{Institute of Discrete Mathematics and Geometry, TU Wien, Vienna, Austria.\newline \texttt{marie-louise.lackner@tuwien.ac.at}} \and Martin Lackner\footnote{Department of Computer Science, 
       		  University of Oxford, Oxford, UK.\newline
              \texttt{martin.lackner@cs.ox.ac.uk}}}
\date{}
\title{On the Likelihood of Single-Peaked Preferences}

\theoremstyle{plain}
\newtheorem{theorem} {Theorem}
\newtheorem{lemma} [theorem] {Lemma}

\newtheorem{proposition} [theorem] {Proposition}
\newtheorem{definition} {Definition}
\newtheorem{corollary} [theorem] {Corollary}

\theoremstyle{definition}
\newtheorem{example} [theorem]{Example}
\newtheorem*{example*}{Example}
\theoremstyle{remark}
\newtheorem{problem*} {Open Problem}
\newtheorem{remark} [theorem]{Remark}

\newcommand{\calT}{{\mathcal{T}}}
\newcommand{\calP}{{\mathcal{P}}}

\newcommand{\cp}{(C,\mathcal{P})}
\newcommand{\GammaSP}{\Gamma_{\text{sp}}}

\newcommand{\ccfont}[1]{\textsf{#1}}

\newcommand{\w}[1]{\ifmmode{\textnormal{\ccfont{W[#1]}}}\else{\textnormal{\ccfont{W[#1]}}}\fi}

\newcommand{\card}[1]{\ensuremath{|#1|}}

\newcommand{\exend}{\ifmmode\hbox{$\dashv$}\else{\unskip\nobreak\hfil\penalty50\hskip1em\null\nobreak\hfil\hbox{$\dashv$}
\parfillskip=0pt\finalhyphendemerits=0\endgraf}\fi}

\newcommand{\defend}{\ifmmode\hbox{$\dashv$}\else{\unskip\nobreak\hfil
\penalty50\hskip1em\null\nobreak\hfil\hbox{$\dashv$}
\parfillskip=0pt\finalhyphendemerits=0\endgraf}\fi}

\newcommand{\vinc}[3]{
\begin{tikzpicture}[baseline = (X.base)]
	\useasboundingbox (0.1,0) rectangle (#1*0.23,0.1);
	\foreach \x/\y in {#2}
	{
		\draw (\x*0.2,0) node (X) {$\y$};
	}
	
	\foreach \z in {#3}
	{
		\ifnum 0<\z
			\ifnum \z<#1
				\draw[thick] (\z*0.2-0.07,-0.19) -- (\z*0.2+0.27,-0.19);
			\fi
		\fi
		
		\ifnum 0=\z
			\draw[thick] (0.07,0.1) -- (0.07,-0.19) -- (0.21,-0.19);
		\fi
		
		\ifnum \z=#1
			\draw[thick] (\z*0.2+0.14,0.1) -- (\z*0.2+0.14,-0.19) -- (\z*0.2,-0.19);
		\fi
	}
\end{tikzpicture}
}

\newcommand{\bivinc}[4]{\!
\begin{tikzpicture}[baseline]

	\foreach \x/\y in {#2}
	{
		\draw (\x*0.2,0.3) node {$\x$};
		\draw (\x*0.2,0) node {$\y$};
	}
	
	\foreach \z in {#3}
	{
		\ifnum 0<\z
			\ifnum \z<#1
				\draw[thick] (\z*0.2-0.07,-0.19) -- (\z*0.2+0.27,-0.19);
			\fi
		\fi
		
		\ifnum 0=\z
			\draw[thick] (0.07,0.1) -- (0.07,-0.19) -- (0.21,-0.19);
		\fi
		
		\ifnum \z=#1
			\draw[thick] (\z*0.2+0.14,0.1) -- (\z*0.2+0.14,-0.19) -- (\z*0.2,-0.19);
		\fi
	}
	
	\foreach \z in {#4}
	{
		\ifnum 0<\z
			\ifnum \z<#1
				\draw[thick] (\z*0.2-0.07,0.49) -- (\z*0.2+0.27,0.49);
			\fi
		\fi
		
		\ifnum 0=\z
			\draw[thick] (0.07,0.21) -- (0.07,0.49) -- (0.21,0.49);
		\fi
		
		\ifnum \z=#1
			\draw[thick] (\z*0.2+0.14,0.21) -- (\z*0.2+0.14,0.49) -- (\z*0.2,0.49);
		\fi
	}
\end{tikzpicture}
\!
}

\def\multiset#1#2{\ensuremath{\left(\kern-.3em\left(\genfrac{}{}{0pt}{}{#1}{#2}\right)\kern-.3em\right)}}

\begin{document}

\maketitle

\begin{abstract}
This paper contains an extensive combinatorial analysis of the single-peaked domain restriction and investigates the likelihood that an election is single-peaked.
We provide a very general upper bound result for domain restrictions that can be defined by certain forbidden configurations.
This upper bound implies that many domain restrictions (including the single-peaked restriction) are very unlikely to appear in a random election chosen according to the Impartial Culture assumption.
For single-peaked elections, this upper bound can be refined and complemented by a lower bound that is asymptotically tight.
In addition, we provide exact results for elections with few voters or candidates.
Moreover, we consider the P\'{o}lya urn model and the Mallows model and obtain lower bounds showing that single-peakedness is considerably more likely to appear for certain parameterizations.
\end{abstract}

\section{Introduction}
The single-peaked restriction~\citep{black} is an extensively studied domain restriction in social choice theory. 
An election, i.e., a collection of preferences represented as total orders on a set of candidates, is single-peaked if the candidates can be ordered linearly---on a so-called axis---so that each preference is either strictly increasing along this ordering, or strictly decreasing, or first increasing and then decreasing. See Figure~\ref{fig:sp-ex} for examples and Section~\ref{sec:preliminaries} for formal definitions.
Intuitively, the axis reflects the society's ordering of the candidates and voters always prefer candidates that are closer to their ideal candidate over those farther away.
In political elections, for example, this axis could reflect the left-right spectrum of the candidates or a natural ordering of political issues such as the maximum income tax.

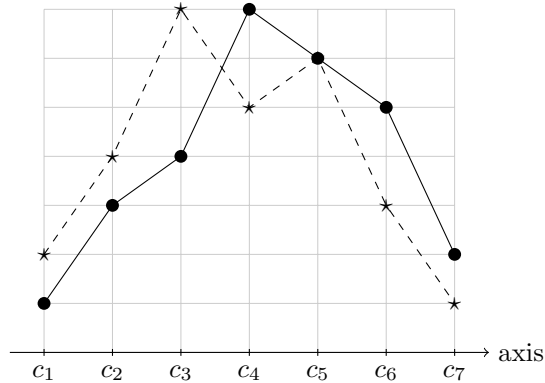
\begin{figure}
\centering
\begin{tikzpicture}[yscale=0.65,xscale=0.9]

  \def\xmin{1}
  \def\xmax{7}
  \def\ymin{0}
  \def\ymax{7}
  
  \draw[step=1cm,black!20,very thin] (\xmin,\ymin) grid (\xmax,\ymax);

  \draw[->] (\xmin -0.5,\ymin) -- (\xmax+0.5,\ymin) node[right] {axis};
  \foreach \x/\xtext in {1/c_1, 2/c_2, 3/c_3, 4/c_4, 5/c_5, 6/c_6, 7/c_7}
    \draw[shift={(\x,\ymin)}] (0pt,2pt) -- (0pt,-2pt) node[below] {$\xtext$};

  \foreach \x/\y in {4/7,5/6,3/4, 6/5, 2/3, 7/2, 1/1}
    \node[fill=black, circle, inner sep=0.6mm] at (\x,\y) {};
    
  \draw[solid,black] (1,1)--(2,3)--(3,4)--(4,7) -- (5,6)--(6,5)--(7,2);

  \foreach \x/\y in {1/2,2/4,3/7, 4/5, 5/6, 6/3, 7/1}
    \node at (\x,\y) {$\star$};
   
   \draw[dashed,black] (1,2)--(2,4)--(3,7)--(4,5) -- (5,6)--(6,3)--(7,1);

\end{tikzpicture}
\label{fig:sp-ex}
\caption{The vote $V_1: c_4> c_5> c_6> c_3> c_2> c_7> c_1$, shown as a solid line, is single-peaked with respect to the axis $c_1<c_2<c_3<c_4<c_5<c_6<c_7$. The vote $V_2: c_3>c_5>c_4>c_2>c_6>c_1>c_7$, depicted as a dashed line, is not single-peaked with respect to this axis since both $c_3$ and $c_5$ form a peak. However, note that both votes are single-peaked with respect to the axis $c_1<c_2<c_3<c_5<c_4<c_6<c_7$.}
\end{figure}

Single-peaked preferences have several nice properties.
First, they guarantee that a Condorcet winner exists and further that the pairwise majority relation is transitive \citep{inada1969simple}.
Thus single-peaked preferences are a way to escape Arrow's paradox \citep{arrow1950difficulty}.
Second, non-manipulable voting rules exist for single-peaked preferences (Moulin 1980) and hence the single-peaked restriction also offers a way to circumvent the Gibbard-Satterthwaite paradox \citep{gib:j:polsci:manipulation,sat:j:polsci:manipulation}.
By adopting an algorithmic viewpoint, a third advantage becomes apparent.
Restricting the input to single-peaked preferences often allows for faster algorithms for computationally hard voting problems \citep{bra-bri-hem-hem:j:single-peaked,BetzlerSU13fully-proportional,wal:c:uncertainty,Faliszewski201189}.

In this paper we perform an extensive combinatorial analysis of the single-peaked domain.
Our aim is to establish results on the likelihood that an election is single-peaked for some axis.
To be more precise, we allow the  axis to be chosen depending on the preferences and do not assume that it is given together with the election.
We consider three probability distributions for elections: the Impartial Culture (IC) assumption in which all total orders are equally likely and are chosen independently, the P\'{o}lya urn model which assumes a certain homogeneity among voters and Mallows model in which the probability of a vote depends on the closeness to a given reference vote (with respect to the Kendall-tau distance).

Our main results are the following:%
\begin{itemize}
\item \emph{Configuration definable restrictions:} Many domain restrictions can be characterized by forbidden configurations, in particular the single-peaked domain.
We prove a close connection between configurations and permutations patterns.
This novel connection allows us to obtain a very general result (Theorem~\ref{thm:SW-bound}), showing that many domain restrictions characterized by forbidden configurations are very unlikely to appear in a random election chosen according to the Impartial Culture assumption.
More precisely, while the total number of elections with $n$ votes and $m$ candidates is equal to $(m!)^n$, the number of elections belonging to such a domain restriction can be bounded by $m! \cdot c^{n m}$ for some constant $c$.
\item \emph{Counting single-peaked elections:} We perform a detailed combinatorial analysis of the single-peaked domain by counting the number of single-peaked elections.
The number of single-peaked elections immediately yields the corresponding probability with respect to the Impartial Culture (IC) assumption, which is the number of single-peaked elections divided by the total number of elections.
We establish an upper bound for the number of single-peaked elections which asymptotically matches our lower bound result (Theorem~\ref{thm:count-sp}).
In addition, we show exact enumeration results for elections with two voters or up to four candidates (Theorem~\ref{thm:SP_small_profiles}).
Our results rigorously show that the single-peaked restriction is highly unlikely to appear in random elections chosen according to the IC assumption.
This holds even for elections with few votes and candidates (cf. Section~\ref{sec:num}).
Most of our results can easily be translated to the Impartial Anonymous Culture (IAC) assumption (Proposition~\ref{thm:iac}).
\item \emph{P\'{o}lya urn model:} In contrast to the IC assumption, single-peaked elections are considerably more likely if elections are chosen according to the P\'{o}lya urn model.
We provide a lower bound on the corresponding likelihood (Theorem~\ref{thm:polya-bound}) and show that, if a sufficiently strong homogeneity is assumed, the probability of an election with $n$ votes being single-peaked is larger than $1/n$ (Corollary~\ref{cor:polya}).
\item \emph{Mallows model:} We encounter the most likely occurrence of single-peaked elections under Mallows model. As for the P\'{o}lya urn model we establish a lower bound result on the likelihood (Theorem~\ref{thm:mallow}).
If the dispersion parameter $\phi$ is sufficiently small%
, we are able to show that single-peaked elections are likely to appear (Corollary~\ref{cor:mallows} and Table~\ref{tab:results-mallows}).
\end{itemize}

\paragraph{Related work.} 

Computing the likelihood of properties related to voting has been the focus of a large body of research.
The most fundamental question in this line of research is the choice of appropriate probability distributions, see the survey of \cite{critchlow_probability_1991}.
We would like to mention two particular properties of elections that have been studied from a probability theoretic point of view: the likelihood of manipulability and the likelihood of having a Condorcet winner.

An election is manipulable if a voter or a coalition of voters is better off by not voting sincerely but by misrepresenting their true preferences.
The Gibbard-Satterthwaite paradox \citep{gib:j:polsci:manipulation,sat:j:polsci:manipulation} states that every reasonable voting rule for more than two candidates is susceptible to manipulation.
However, the Gibbard-Satterthwaite paradox does not offer insight into how likely it is that manipulation is possible.
Determining this likelihood both for single manipulators and coalitions of manipulators has been the focus of intensive research.
Results have been obtained under a variety of probability distributions:
for example under the Impartial Culture assumption \citep{slinko_asymptotic_2002,slinko_asymptotic_2002a,fri-kal-nis:c:quantiative-gib-sat,isaksson2012geometry}, the P\'{o}lya urn model \citep{lepelley2003voting}, the Impartial Anonymous Culture \citep{favardin2002borda,slinko_how_2005}.

The likelihood that an election has a Condorcet winner or, its converse, the likelihood of the Condorcet paradox has been the focus of many publications
(see the survey of \cite{condorcet-paradox} as well as more recent work of  \cite{gehrlein2013impact,condorcet-diversity} for more recent research).
In particular, we would like to mention that the likelihood of an election with three candidates having a Condorcet winner under the Impartial Anonymous Culture assumption is $\frac{15(n+3)^2}{16(n+2)(n+4)}$ for odd $n$ and $\frac{15(n+2)(n^2+8n+8)}{16(n+1)(n+3)(n+5)}$ for even $n$ \citep{gehrlein02}.
We will comment on the relation between these result and our results in Section~\ref{sec:concl}.

\paragraph{Organization.} Preliminaries are established in Section~\ref{sec:preliminaries}. The results on configuration definable restrictions can be found in Section~\ref{sec:stanley wilf}, results on counting single-peaked elections in Section~\ref{sec:IC}, results on the IAC assumption in Section~\ref{sec:iac}, results on the P\'{o}lya urn model in Section~\ref{sec:polya} and results on the Mallows model in Section~\ref{sec:mallow}.
In Section~\ref{sec:num} we provide numerical evaluations of our results and discuss their implications.
We conclude the paper in Section~\ref{sec:concl} with directions for future research.

\section{Preliminaries}
\label{sec:preliminaries}

\paragraph{Sets and orders.}
Let $S$ be a finite set.
A relation on $S$ is total if for every $a,b\in S$, either the pair $(a,b)$ or $(b,a)$ is contained in the relation.
A \emph{total order on $S$} is a reflexive, antisymmetric, transitive and total relation.
Let $T$ be a total order of $S$.
Instead of writing $(a,b)\in T$, we write $a\leq_T b$ or $b\geq_T a$.
We write $a <_T b$ or $b >_T a$ to state that $a\leq_T b$ and $a\neq b$.
As a short form, we write $T:s_1 s_2 s_3 \ldots s_i$ instead of $s_1 >_T s_2 >_T s_3>_T \ldots>_T  s_i$ for $s_1, s_2, \ldots, s_i$ in $S$.
We write $T(i)$ to denote the $i$-th largest element with respect to $T$.

\paragraph{Permutations.}
A permutation $\pi$ of a finite set $S$ is a bijective function from $S$ to $S$.
We write $\pi^{-1}$ for the inverse function of $\pi$.
A permutation of the set $\{1,\ldots,m\}$ is called an $m$-permutation.
We shall write an $m$-permutation $\pi$ as the sequence of values $\pi(1) \pi(2) \ldots \pi(m)$.
For example $\pi=321$ is the permutation with $\pi(1)=3, \pi(2)=2$ and $\pi(3)=1$.
Every pair $(T_1, T_2)$ of total orders on a set with $m$ elements yields an $m$-permutation $p(T_1, T_2)$, which is defined as follows:
$i$ maps to $j$ if the $i$-th largest element in $T_1$ equals the $j$-th largest element in $T_2$.
For $T_1:bac$ and $T_2:cab$ we have $p(T_1, T_2)=321$.
Note that $p(T_1,T_2)=p(T_2,T_1)^{-1}$.

\paragraph{Elections.}
An $(n,m)$-election $\cp$ consists of a size-$m$ set $C$ and an $n$-tuple $\mathcal{P}=(V_1,\ldots,V_n)$ of total orders on $C$.
The total orders $V_1,\ldots,V_n$ represent votes or preferences.
We write $V\in \mathcal{P}$ to denote that there exists an index $i\in[n]$ such that $V=V_i$.
Given a vote $V_i\in\calP$, the intuitive meaning of $V_i: c_j c_k$ is that the $i$-th voter prefers candidate $c_j$ to candidate $c_k$.

We assume that candidate sets are chosen from a fixed, infinite set.
When counting elections we do not care about the specific names these candidates have.
That means when we count elections we fix the candidate set to $\{c_1, c_2$, $\ldots, c_m\}$.
Note that the number of $(n,m)$-elections is $(m!)^n$.
Throughout the paper we only consider  $(n,m)$-elections with $n\geq2$ and $m\geq2$.

\paragraph{Probability distributions over elections.}
We consider four probability distributions in this paper.
The first and simplest is the Impartial Culture (IC) which assumes that in an election all votes, i.e., total orders of candidates, are equally likely and are chosen independently.
Thus, the IC assumption can be seen as the uniform distribution over total orders on candidates.
The results in our paper concerning IC do not state probabilities but rather count the number of elections.
If, e.g., the number of  single-peaked $(n,m)$-elections is $a(n,m,SP)$, then the probability under the IC assumption that an $(n,m)$-election is single-peaked is $\frac{a(n,m,SP)}{(m!)^n}$.
It is important to note that in our paper elections contain an ordered list of votes. Thus, we distinguish elections that consist of the same votes but these votes appear in a different order.
This is in contrast to the Impartial Anonymous Culture (IAC) assumption, in which elections contain a multiset of votes and thus elections are not ordered. The IAC assumption is briefly considered in Section~\ref{sec:iac}.

In addition to the IC assumption, we consider the P\'{o}lya urn model and the Mallows model. Both distributions are generalizations of the IC assumption and generate more structured elections.
We are going to define the P\'{o}lya urn and the Mallows model in Section~\ref{sec:polya} and \ref{sec:mallow}, respectively.

\paragraph{Single-peaked preferences.}

The single-peaked restriction 
assumes that the candidates can be ordered linearly on a so-called axis 
and voters prefer candidates close to their ideal point to candidates that are further away. 
\begin{definition}
Let $\cp$ be an election and $A$  a total order of $C$.
A vote $V$ on $C$ \emph{contains a valley with respect to $A$ on the candidates $c_1,c_2,c_3\in C$} if $A: c_1\, c_2\, c_3$ and $V$ ranks $c_2$ below $c_1$ and $c_3$.
The election $\cp$ is \emph{single-peaked with respect to $A$} if for every $V\in\mathcal{P}$ and for all candidates $c_1,c_2,c_3\in C$, $V$ does not contain a valley with respect to $A$ on $c_1,c_2,c_3$.
We then call the total order $A$ the axis.
The election $\cp$ is \emph{single-peaked} if there exists a total order $A$ of $C$ such that $\cp$ is single-peaked with respect to $A$.
\end{definition}

\begin{remark}
Given an axis on $m$ candidates, there are $2^{m-1}$ votes that are single-peaked with respect to this axis \citep{esc-lan-oez:c:single-peak}. This can be seen as follows:
The last ranked candidate has to be one of the two outermost candidates on the axis and hence there are two possibilities.
Once we have picked this last candidate, we can iterate the argument for the next lowest ranked candidate, where we again have two possibilities.
Thus, for all positions in the total order (except for the top ranked candidate), there are two candidates to choose from---which yields $2^{m-1}$ possibilities in total.%
\label{rem:number-of-sp-votes}
\end{remark}

\section{A general result based on permutation patterns}
\label{sec:stanley wilf}

Before we study the single-peaked domain in detail, we prove a general result that is applicable to a large class of domain restrictions including the single-peaked domain.
To precisely define this class of domain restrictions, we require the notion of configuration definability.

\subsection{Configuration definable domain restrictions} 

Single-peaked elections may also be defined in the following way:

\begin{theorem}[\citealp{bal-hae:j:characterization-single-peaked}]
\label{thm:bal-hae-sp}
An $(n,m)$-election $(C,\calP)$ is single-peaked if and only if
there do not exist candidates $a,b,c,d\in C$ and indices $i,j\in[n]$ such that
\begin{itemize}
\item $V_i: abc$, $V_i: db$, 
\item $V_j: cba$ and $V_j: db$ holds
\end{itemize}
and there do not exist candidates $a,b,c\in C$ and indices $i,j,k\in[n]$ such that
\begin{itemize}
\item $V_i: ba, ca$ (i.e., $a$ is ranked below $b$ and $c$),
\item $V_j: ab, cb$ and 
\item $V_k: ac, bc$ holds.
\end{itemize}
\end{theorem}

Note that this theorem defines single-peakedness without referring to an axis.
Indeed, single-peakedness is now defined as a local property in the sense that certain \emph{configurations} must not be contained in the election.
Similar definitions have also been found for the single-crossing~\citep{DBLP:journals/scw/BredereckCW13} and group-separable~\citep{bal-hae:j:characterization-single-peaked} domain.
For other domains such as value-, worst-, medium and best-restricted preferences, a characterization via configurations follows immediately from the original definitions \citep{sen:j:possibility,sen-pat:j:majority}.
Let us now exactly define what it means for a domain restriction to be \emph{configuration definable}.

\begin{definition}%
An $(l,k)$-\emph{configuration $(S,\mathcal{T})$} consists of a finite set $S$ of cardinality $k$ and a tuple $\mathcal{T}=(T_1,\ldots,T_l)$, where $T_1,\ldots,T_l$ are total orders on $S$.
An election $(C,\mathcal{P})$ \emph{contains} configuration $\mathcal{C}$
if there exist an injective function $f$ from $[l]$ into $[n]$
and an injective function $g$ from $S$ into $C$
such that, for any $x,y\in S$ and $i\in [l]$, it holds that
$T_i: xy$ implies $V_{f(i)}: g(x)g(y)$.
\label{def:configuration-containment}
\end{definition}

We use $(S,\mathcal{T})\sqsubseteq(C,\mathcal{P})$ as a shorthand notation to denote that the election $(C,\mathcal{P})$ contains the configuration $(S,\mathcal{T})$.
An election $(C,\mathcal{P})$ \emph{avoids} a configuration $(S,\mathcal{T})$ if $(C,\mathcal{P})$ does not contain $(S,\mathcal{T})$.
In such a case we say that $(C,\mathcal{P})$ is \emph{$(S,\mathcal{T})$-restricted}.
If the set $S$ is clear from the context, we omit it and just use $\mathcal{T}$ to describe a configuration.

\begin{example}
Let us consider an election $(C,\mathcal{P})$ with $C=\{u,v,w,x,y\}$ and $\calP=(uvwxy,wyvux,yuxwv)$ and 
a configuration $(S,\mathcal{T})$ with $S=\{a,b,c,d\}$ and $\calT=(dabc,cdba)$.
Election $(C,\mathcal{P})$ contains the configuration $(S,\mathcal{T})$ as witnessed by the functions $f:\{1\mapsto 1, 2\mapsto 3\}$ and $g:\{a\mapsto v, b\mapsto x, c\mapsto y, d\mapsto u\}$.
In Figure~\ref{fig:ex:isomorph}, the functions $f$ and $g$ are depicted graphically.
\begin{figure}
\begin{center}
\usetikzlibrary{shapes,snakes}

\begin{tikzpicture}[scale=0.8]
\node[] at (2,1.2)  {Configuration:};
\node[draw, thick,rectangle, inner sep=7pt] (c00) at (0,0)  {d};
\node[] (c00) at (1.4*0.5,0)  {$>$};
\node[draw, thick,regular polygon,regular polygon sides=5] (c10) at (1.4*1,0)  {a};
\node[] (c00) at (1.4*1.5,0)  {$>$};
\node[draw, thick,circle] (c20) at (1.4*2,0)  {b};
\node[] (c00) at (1.4*2.5,0)  {$>$};
\node[draw, thick,diamond] (c30) at (1.4*3,0)  {c};
\node[draw, thick,diamond] (c01) at (0,-1)  {c};
\node[draw, thick,rectangle, inner sep=7pt] (c11) at (1.4*1,-1)  {d};
\node[draw, thick,circle] (c21) at (1.4*2,-1)  {b};
\node[draw, thick,regular polygon,regular polygon sides=5] (c31)  at (1.4*3,-1) {a};
\node[] (c00) at (1.4*0.5,-1)  {$>$};
\node[] (c00) at (1.4*1.5,-1)  {$>$};
\node[] (c00) at (1.4*2.5,-1)  {$>$};

\draw[draw,->,thick]  (1.4*3.5,0) --(7,0.5) node [midway, sloped,  above, fill=white] {f};
\draw[->,thick]  (1.4*3.5,-1) --(7,-1.5) node [midway, sloped,  above, fill=white] {f};

\node[] at (10.5,1.5)  {Election:};
\node[draw, thick,rectangle, inner sep=7pt] (p00) at (1.4*8-3.5,0.5)  {u};
\node[draw, thick,regular polygon,regular polygon sides=5] (p10) at (1.4*9-3.5,0.5)  {v};
\node at (1.4*10-3.5,0.5)  {w};
\node[draw, thick,circle]  at (1.4*11-3.5,0.5)  {x};
\node[draw, thick,diamond]at (1.4*12-3.5,0.5)  {y};
\node at (1.4*8-3.5,-0.5)  {w};
\node at (1.4*9-3.5,-0.5)  {y};
\node at (1.4*10-3.5,-0.5)  {v};
\node at (1.4*11-3.5,-0.5)  {u};
\node at (1.4*12-3.5,-0.5)  {x};
\node[draw, thick,diamond]at (1.4*8-3.5,-1.5)  {y};
\node[draw, thick,rectangle, inner sep=7pt] (p00) at (1.4*9-3.5,-1.5)  {u};
\node[draw, thick,circle]  at (1.4*10-3.5,-1.5)  {x};
\node at (1.4*11-3.5,-1.5)  {w};
\node[draw, thick,regular polygon,regular polygon sides=5] at (1.4*12-3.5,-1.5)  {v};

\node[] (c00) at (1.4*8.5-3.5,0.5)  {$>$};
\node[] (c00) at (1.4*9.5-3.4,0.5)  {$>$};
\node[] (c00) at (1.4*10.5-3.5,0.5)  {$>$};
\node[] (c00) at (1.4*11.45-3.5,0.5)  {$>$};
\node[] (c00) at (1.4*8.5-3.5,-0.5)  {$>$};
\node[] (c00) at (1.4*9.5-3.4,-0.5)  {$>$};
\node[] (c00) at (1.4*10.5-3.5,-0.5)  {$>$};
\node[] (c00) at (1.4*11.45-3.5,-0.5)  {$>$};
\node[] (c00) at (1.4*8.5-3.45,-1.5)  {$>$};
\node[] (c00) at (1.4*9.5-3.45,-1.5)  {$>$};
\node[] (c00) at (1.4*10.5-3.45,-1.5)  {$>$};
\node[] (c00) at (1.4*11.45-3.5,-1.5)  {$>$};

\end{tikzpicture}
\end{center}%
\caption{The configuration on the left-hand side is contained in the election on the right-hand side as witnessed by $f:\{1\mapsto 1, 2\mapsto 3\}$ and $g:\{a\mapsto v, b\mapsto x, c\mapsto y, d\mapsto u\}$.}%
\label{fig:ex:isomorph}
\end{figure}
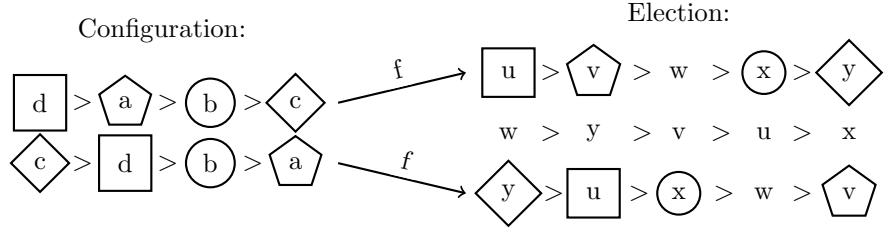
\end{example}

By considering all linearizations of the partial orders appearing in Theorem~\ref{thm:bal-hae-sp} we can now restate it as follows.

\begin{theorem}[\citealp{bal-hae:j:characterization-single-peaked}]
An election is single-peaked if and only if it avoids 
\begin{itemize}
\item the following $(2,4)$-configurations:

$(dabc,dcba)$,
$(adbc,dcba)$,
$(dabc,$ $cdba)$ and  
$(adbc,cdba)$ 
\item as well as the following $(3,3)$-configurations:

$(bca,acb,$ $abc)$,
$(cba,acb,abc)$,
$(bca,cab,abc)$,
$(cba$, $cab$, $abc)$,
$(bca,acb,bac)$,
$(cba,$ $acb,$ $bac)$,
$(bca,$ $cab,$ $bac)$,
$(cba,cab,bac)$.
\end{itemize}
\label{thm:sp-confs}
\end{theorem}
The first four configurations correspond to the first condition in Theorem~\ref{thm:bal-hae-sp}, the remaining eight correspond to the second condition.

\begin{definition}%
Let $\Gamma$ be a set of configurations.
A set of elections $\Pi$ is \emph{defined by $\Gamma$} if $\Pi$ consists exactly of those elections that avoid all configurations in $\Gamma$.
We call $\Pi$ \emph{configuration definable} if there exists a set of configurations $\Gamma$ which defines $\Pi$.
If $\Pi$ is definable by a finite set of configurations, it is called \emph{finitely configuration definable}.
\end{definition}
By Theorem~\ref{thm:sp-confs} we know that the set of all single-peaked elections is finitely configuration definable.
This is also true for the set of group-separable elections~\citep{bal-hae:j:characterization-single-peaked} and for the set of single-crossing elections~\citep{DBLP:journals/scw/BredereckCW13}.

We are now going to characterize which sets of elections are configuration definable.
In the following definition, for two elections $\cp$ and $(C',\mathcal{P}')$, we write $(C',\mathcal{P}') \mathbin{\sqsubseteq} (C,\mathcal{P})$ if $(C',\mathcal{P}')$, considered as a configuration, is contained in $\cp$.
Since every election can be seen as a configuration, the configuration containment relation immediately translates to election containment.

\begin{definition}%
A set of elections $\Pi$ is \emph{hereditary} if for every election $(C',\mathcal{P}')$ it holds that if there exists an election $(C,\mathcal{P})\in\Pi$ with $(C',\mathcal{P}')\sqsubseteq (C,\mathcal{P})$, then $(C',\mathcal{P}')\in \Pi$.%
\label{def:hereditary}
\end{definition}

\begin{proposition}
A set of elections is configuration definable if and only if it is hereditary.%
\label{prop:inf-conf-definable}
\end{proposition}

\begin{proof}
Let a set of elections  $\Pi$ be defined by a set of configurations $\Gamma$ and $(C,\mathcal{P})\in\Pi$.
Let $(C',\mathcal{P}')\sqsubseteq (C,\mathcal{P})$.
Since $(C,\mathcal{P})\in\Pi$, $(C,\mathcal{P})$ avoids all configurations in $\Gamma$.
Due to $(C',\mathcal{P}')\sqsubseteq (C,\mathcal{P})$, also $(C',\mathcal{P}')$ avoids all configurations in $\Gamma$ and is therefore contained in $\Pi$.

For the other direction,
let $\Pi$ be a hereditary set of elections.
We define $\Pi^c$ to be the \emph{complement of $\Pi$}, i.e., an election is contained in $\Pi^c$ if it is not contained in $\Pi$.
We claim that $\Pi^c$, considered as a set of configurations, defines $\Pi$.
If $(C,\mathcal{P})\in \Pi$, it avoids all $(C',\mathcal{P}')\in\Pi^c$ since if there existed a $(C',\mathcal{P}')\in\Pi^c$ with $(C',\mathcal{P}')\sqsubseteq (C,\mathcal{P})$, this would imply that $(C',\mathcal{P}')\in\Pi$.
It remains to show that if an election $(C'',\mathcal{P}'')$ avoids all $(C',\mathcal{P}')\in\Pi^c$, then $(C'',\mathcal{P}'')\in\Pi$.
This follows from noting that $(C'',\mathcal{P}'')$ avoiding all $(C',\mathcal{P}')\in\Pi^c$ implies $(C'',\mathcal{P}'')\notin \Pi^c$, which in turn implies $(C'',\mathcal{P}'')\in \Pi$.
\end{proof}
As a consequence of Proposition~\ref{prop:inf-conf-definable}, we know that 2D single-peaked \citep{barbera1993generalized} and 1D Euclidean elections \citep{theoryofdata,Knoblauch} are configuration definable.
However, Proposition~\ref{prop:inf-conf-definable} does not help to answer whether these restrictions are \emph{finitely} configuration definable.
For the 1D Euclidean domain it is even known that it is not finitely configuration definable~\citep{1d-infinte}.
Finite configuration definability has been crucial for establishing algorithmic results~\citep{bredereck-nearly,edith-approx}.

A natural example of a meaningful restriction that is not configuration definable is the set of all elections that have a Condorcet winner. The property of having a Condorcet winner is not hereditary and thus cannot be defined by configurations.
Another example is the ``single-peaked on a tree'' restriction~\citep{dem:j:sp-tree}.

\subsection{The connection to permutation patterns}

In this section, we establish a strong link between the concept of configuration containment and the concept of pattern containment in permutations.
Pattern containment in permutations is defined as follows.

\begin{definition}
A $k$-permutation $\pi$ is \emph{contained as a pattern} in an $m$-permu\-tation $\tau$ if there is a subsequence of $\tau$ that is order-isomorphic to $\pi$.
In other words, $\pi$ is contained in $\tau$, if there is a strictly increasing map $\mu: \{1,\ldots,k\} \rightarrow \{1,\ldots,m\}$ so that the sequence $\mu(\pi)=\big(\mu(\pi(1)),\mu(\pi(2)),\ldots,\mu(\pi(k))\big)$ is a subsequence of $\tau$.
This map $\mu$ is called a \emph{matching} of $\pi$ into $\tau$.
If there is no such matching, \emph{$\tau$ avoids the pattern $\pi$}.%
\label{def:pattern-containment}
\end{definition}

For example, the pattern $\pi=132$ is contained in $\tau=32514$ since the subsequence $254$ of $\tau$ is order-isomorphic to $\pi$.
However, the pattern $123$ is avoided by $\tau$. 
Note that $\tau$ contains $\pi$ if and only if $\tau ^{-1}$ contains $\pi ^{-1}$.

We are going to prove two lemmas.
The first lemma (Lemma~\ref{lem:confavoiding<->pp}) states that every permutation pattern matching query can naturally be translated into a configuration containment query.
The second lemma (Lemma~\ref{lem:enumeration_confavoiding<->pp}) states that for $(2,k)$-configurations, a configuration containment query can naturally be translated in a permutation pattern query.

\begin{lemma}
Let $\pi$ be a $k$-permutation and $\tau$  an $m$-permu\-tation.
We define the corresponding configuration and election as follows:
Let $\cp$ be a $(3,m)$-election with $C=\{c_1,\ldots,c_m\}$, $\mathcal{P}=(V_1, V_2, V_3)$ and
\begin{align*}
 V_1:c_1 c_2  \cdots  c_m \qquad
 V_2:c_1 c_2  \cdots  c_m \qquad
 V_3:c_{\tau(1)}  c_{\tau(2)}  \cdots  c_{\tau(m)}.
\end{align*}
Furthermore, let $(S,\calT)$ be a $(3,k)$-configuration with $S=\{x_1,\ldots,x_k\}$, $\mathcal{T}=(T_1, T_2, T_3)$ and 
\begin{align*}
T_1:x_1 x_2  \cdots  x_k \qquad T_2:x_1 x_2  \cdots  x_k \qquad T_3:x_{\pi(1)} x_{\pi(2)} \cdots x_{\pi(k)}.
\end{align*}
Then $\cp$ contains $(S,\calT)$ if and only if $\tau$ contains $\pi$.%
\label{lem:confavoiding<->pp}
\end{lemma}

\begin{proof}
Assume that we have a matching $\mu$ from $\pi$ into $\tau$.
We have to find an injective function $f$ from $\{1,2,3\}$ into $\{1,2,3\}$
and an injective function $g$ from $S$ into $C$
such that, for any $x,y\in S$ and $i\in \{1,2,3\}$, it holds that
$T_i: xy$ implies $V_{f(i)}: g(x)g(y)$.
Let $f$ be the function $\{1\mapsto 1, 2\mapsto 2, 3\mapsto 3\}$ and $g=\mu$.
It holds for $x_i,x_j \in S$ that $T_1: x_i\,x_j$ if and only if $V_1: c_{\mu(i)}\, c_{\mu(j)}$ since $\mu$ is monotone.
The same holds for $T_2$ and $V_2$.
For $T_3$ and $V_3$ observe that $T_3: x_i\,x_j$ implies $V_3: c_{\mu(i)}\, c_{\mu(j)}$ since $\mu$ is a matching.
Thus, the election fulfils $(S,\calT)\sqsubseteq \cp$.

For the other direction, assume that $\cp$ contains $(S,\calT)$.
Consequently, there exists an injective function $f$ from $\{1,2,3\}$ into $\{1,2,3\}$
and an injective function $g$ from $S$ into $C$
such that, for any $x,y\in S$ and $i\in \{1,2,3\}$, it holds that
$T_i: xy$ implies $V_{f(i)}: g(x)g(y)$.
First, we claim that $f(3)=3$.
Observe that $f$ has to map $T_1$ and $T_2$ to identical total orders.
Thus, unless $V_1=V_2=V_3$, $f(3)=3$.
In the case that $V_1=V_2=V_3$, we can assume without loss of generality that $f(3)=3$.
We will construct a function $\mu$ and show that $\mu$ is a matching from $\pi$ into $\tau$.
Let us define $\mu(i)=j$ if $g(x_i)=c_j$.
Observe that $\mu$ is strictly increasing since for $i<j$, $V_1: c_{g(i)}\,c_{g(j)}$ and $V_1:c_1 c_2  \cdots  c_m$.
In addition, $\mu(\pi)=\big(\mu(\pi(1)),\mu(\pi(2)),\ldots,\mu(\pi(k))\big)$ is a subsequence of $\tau$ since, by definition of $T_3$ and $V_3$ and the fact that $f(3)=3$, $\big(g(x_{\pi(1)}),g(x_{\pi(2)}),\ldots,g(x_{\pi(k)})\big)$ is a subsequence of $\big(c_{\tau(1)},c_{\tau(2)},\ldots,c_{\tau(m)}\big)$.
\end{proof}

Next, we will prove the second lemma, which is essential for the main theorem of this section (Theorem~\ref{thm:SW-bound}).
As of now, we shall denote by $S_m(\pi_1,\ldots,\pi_l)$ the cardinality of the set of $m$-permutations that avoid the patterns $\pi_1,\ldots,\pi_l$. 

\begin{lemma}
Let $(S,\calT)$ be a $(2,k)$-configuration with $\calT=(T_1,T_2)$.
Furthermore, let $V_1$ be a total order on the candidate set $C=\{c_1,\ldots,c_m\}$.
Then the number of total orders $V_2$ such that the election $\cp$ with $\mathcal{P}=( V_1, V_2 )$ avoids $(S,\calT)$ is equal to $S_m(\pi, \pi ^{-1})$, where $\pi=p(T_1,T_2)$.%
\label{lem:enumeration_confavoiding<->pp}
\end{lemma}

\begin{proof}
Let us start by proving the following statement: 
The configuration $(S,\calT)$ is contained in an election $\cp$ with $\mathcal{P}=(V_1,V_2)$ if and only if the permutation $\pi$ or the permutation $\pi^{-1}$ is contained in $p(V_1,V_2)$.
In order to alleviate notation, we will assume in the following that $C=\{1, 2, \ldots m\}$ and $S=\{1, 2, \ldots k\}$.

Let $\mu$ be a matching witnessing that $\pi$ is contained in $p(V_1,V_2)$.
We can assume without loss of generality that $T_1: 12\ldots k$ and $V_1:12 \ldots m$.
Then the functions $f=\{1\mapsto 1, 2\mapsto 2\}$ and $g=\mu$ show that $(S,\calT)\sqsubseteq\cp$ (cf.\ Definition~\ref{def:configuration-containment}).
If $\pi^{-1}$ is contained in $p(V_1,V_2)$ as witnessed by a matching $\mu$, then the functions $f=\{1\mapsto 2, 2\mapsto 1\}$ and $g=\mu$ show that $(S,\calT)\sqsubseteq\cp$.

For the other direction, let $(S,\calT)\sqsubseteq\cp$.
Without loss of generality we assume that $T_1:12 \ldots k$.
Note that renaming $C$ does not change whether $(S,\calT)\sqsubseteq\cp$.
Thus, it is safe to rename the candidates according to the $f$ function:
If $f=\{1\mapsto 1, 2\mapsto 2\}$, let $V_1:12 \ldots n$.
Since $f(1)=1$, $g$ is monotonic.
It is easy to verify that $g$ is a matching from $\pi$ into $p(V_1,V_2)$.
If $f=\{1\mapsto 2, 2\mapsto 1\}$, 
let $V_2:12 \ldots n$.
Now, $g$ is a matching from $\pi$ into $p(V_2,V_1)=(p(V_1,V_2))^{-1}$.
This is equivalent to $g$ being a matching from $\pi^{-1}$ into $p(V_1,V_2)$.

It follows that $\cp$ avoids the configuration $(S,\calT)$ if and only if the permutation $p(V_1,V_2)$ avoids both the patterns $\pi$ and $\pi^{-1}$.
Moreover, for the fixed total order $V_1$ and a fixed $m$-permutation $\tau$, there is a single total order $V_2$ such that $p(V_1,V_2)=\tau$.
Thus the number of votes $V_2$ such that $p(V_1,V_2)$ avoids $\pi$ and $\pi ^{-1}$ (and equivalently the number of votes $V_2$ such that $\cp$ avoids $(S,\calT)$) is equal to $S_m(\pi, \pi ^{-1})$, the number of $m$-permutations avoiding $\pi$ and $\pi ^{-1}$.
\end{proof}
From this lemma follows the main theorem of this section that is applicable to any set of configurations that contains at least one configuration of cardinality two.

\subsection{Elections that avoid a $(2,k)$-configuration}

With the help of Lemma~\ref{lem:enumeration_confavoiding<->pp}, we are able to establish the following result.

\begin{theorem}
Let $a(n,m,\Gamma)$ be the number of $(n,m)$-elections avoiding a set of configurations $\Gamma$.
Let $k\geq 2$.
If a set of configurations $\Gamma$ contains a $(2,k)$-configuration, then it holds for all $n, m \in \mathbb{N}$ that
\[a(n,m,\Gamma) \leq m!\cdot c_k^{(n-1)m},\]
where $c_k$ is a constant depending only on $k$.%
\label{thm:SW-bound}
\end{theorem}

This result shows that forbidding any $(2,k)$-configuration is a very strong restriction.
Indeed, $m!\cdot c_k^{(n-1)m}$ is very small compared to the total number of $(n,m)$-elections which is $(m!)^n$.
This result allows us to bound the number of single-peaked and group-separable elections.
However, let us prove this result first before we explore its consequences. 

In order to prove this result we make use of the link between configuration avoiding elections and pattern avoiding permutations established in Lemma~\ref{lem:enumeration_confavoiding<->pp} and profit from a very strong result within the theory of pattern avoidance in permutations, the Marcus-Tardos theorem (former Stanley-Wilf conjecture).

\begin{proof}%
We are going to provide an upper bound on the number of $(n,m)$-elections avoiding a $(2,k)$-configuration $(S,\calT)$ with $\calT=(T_1,T_2)$.
Let us start by choosing the first vote $V_1$ of the election at random.
For this there are $m!$ possibilities.
When choosing the remaining $(n-1)$ votes $V_2, \ldots, V_n$, we have to make sure that no selection of two votes contains the forbidden configuration $(S,\calT)$.
We relax this condition and only demand that for none of the pairs $(V_1, V_i)$ with $i \neq 1$, the election $(C,(V_1,V_i))$ contains the forbidden configuration.
Hereby we obtain an upper bound for $a(n,m,\left\lbrace (S,\calT)\right\rbrace)$.
Now Lemma~\ref{lem:enumeration_confavoiding<->pp} tells us that there are---under this relaxed condition---$S_m(\pi, \pi ^{-1})$ choices for every $V_i$ where $\pi = p(T_1,T_2)$. 
Thus we have the following upper bound:
\begin{align}
a(n,m,\{ (S,\calT) \}) \leq m!  S_m(\pi, \pi ^{-1})^{n-1}  \leq m! S_m(\pi)^{n-1},
\label{eqn:SW-inequality}
\end{align}
where the second inequality follows since all permutations avoiding both $\pi$ and $\pi ^{-1}$ clearly avoid $\pi$.

Now we apply the famous Marcus-Tardos theorem~\citep{marcus2004excluded}: For every permutation $\pi$ of length $k$ there exists a constant $c_k$ such that for all positive integers $m$ we have $S_m(\pi) \leq {c_k}^m$.
Putting this together with Equation~\eqref{eqn:SW-inequality} and noting that $a(n,m,\left\lbrace (S,\calT) \right\rbrace )$ is an upper bound for $a(n,m,\Gamma)$,  we obtain the desired upper bound.
\end{proof}
The proof of the Marcus-Tardos theorem provides an explicit exponential formula for the constants $c_k$.
Indeed, it holds that  \[
S_m(\pi) \leq \left( 15^{2k^4 \binom{k^2}{k}}\right)^m.
\] 
These constants are however far from being optimal and there is an ongoing effort to find minimal values for $c_k$ with fixed $k$.
In particular it has been shown that $c_2=1$, $c_3=4$
\citep{simion1985restricted} and $c_4 \leq 13.738$ \citep{bona2014new}. 

Let us discuss the implications of this theorem.
It is applicable to all (not necessarily finite) configuration definable domain restrictions that contain a configuration of cardinality two.
In particular, we obtain the following upper bounds for  single-peaked and  group-separable elections.

\begin{corollary}
Let $a(n,m,\Gamma_{sp})$ denote the number of  single-peaked $(n,m)$-elections.
For $n,m \geq 2$ it holds that $a(n,m, \Gamma_{sp}) \leq m!\cdot 4^{(m-1)(n-1)}$.%
\label{cor:count-single-peaked}
\end{corollary}

\begin{proof}
We know from Theorem~\ref{thm:bal-hae-sp} that the single-peaked domain avoids the $(2,4)$-configurations $(dabc,dcba)$,
$(adbc,dcba)$,
$(dabc,cdba)$ and  
$(adbc,cdba)$.
We  can use Equation~\eqref{eqn:SW-inequality} in the proof of Theorem~\ref{thm:SW-bound} to bound $a(n,m,\Gamma_{sp})$.
For this, we have to compute the permutations and their inverses corresponding to the four configurations.
We obtain the permutations $\pi_1=p(dabc,dcba)=1432$, $\pi_2=p(adbc,dcba)=4132$, $\pi_3=p(dabc,cdba)=2431$ and $\pi_4=p(adbc,cdba)=4231$.
Their inverses are $\pi_1^{-1}=\pi_1$, $\pi_2^{-1}=\pi_3$, $\pi_3^{-1}=\pi_2$ and $\pi_4^{-1}=\pi_4$.
Hence it holds that the number of  $(n,m)$-elections that avoid these four configurations is bounded by
$m! \cdot S_m(\pi_1, \pi_2, \pi_3, \pi_4)^{n-1}$.
The enumeration problem for this permutation class has been solved by \cite{guibert1995permutations} in his PhD thesis with the help of the method of generating trees.
A more direct and combinatorial approach to this permutation class can be found in the first author's PhD thesis \citep{bruner2015patterns}.
It holds that $S_m(\pi_1, \pi_2, \pi_3, \pi_4)=\binom{2m-2}{m-1}$, which, in turn, is bounded by $4^{m-1}$.
\end{proof}
This upper bound also holds for the 1D Euclidean domain \citep{theoryofdata,Knoblauch}, since this domain is a subset of the single-peaked domain.
In the next section, we will see that the growth rate of $a(n,m, \Gamma_{sp})$ is indeed of the form $m!\cdot c^{(m-1)(n-1)}$ for some constant $c$.
However, the constant found in Corollary~\ref{cor:count-single-peaked} is not optimal as we will see by providing a better bound for the single-peaked restriction that is even asymptotically optimal.

As another corollary of Theorem~\ref{thm:SW-bound}, we prove a bound on the number of  group-separable elections.
An election is group separable if for every subset of candidates $C'$ there exists a partition $C_1,C_2$ of $C'$ such that in every vote either all candidates in $C_1$ are preferred to all candidates in $C_2$ or vice versa.
\citet{bal-hae:j:characterization-single-peaked} showed that the group-separable domain is finitely configuration definable. In particular, this domain avoids the configuration $(abcd,bdac)$.
Therefore, Theorem~\ref{thm:SW-bound} is applicable.

\begin{corollary}
Let $a(n,m,\Gamma_{gs})$ denote the number of  group-separable $(n,m)$-elections.
For $n,m \geq 2$ it holds that $a(n,m, \Gamma_{gs}) \leq m!\cdot (3 +2\sqrt{2})^{m(n-1)}$.%
\label{cor:count-group-separable}
\end{corollary}

\begin{proof}
The proof is similar to the one of Corollary~\ref{cor:count-single-peaked}.
We use Equation~\eqref{eqn:SW-inequality} in the proof of Theorem~\ref{thm:SW-bound} to bound $a(n,m,\Gamma_{gs})$, i.e.,
$a(n,m,\Gamma_{gs})\leq m! \cdot S_m(\pi, \pi ^{-1})^{n-1}$, where $\pi = p(abcd,bdac) = 3142$
and $\pi ^{-1} = 2413$.
Permutations avoiding these two patterns are known under the name of \emph{separable permutations}.
It is known that separable permutations are counted by the large Schr\"oder numbers (OEIS A006318) and that
$S_m(\pi, \pi ^{-1}) \leq (3 +2\sqrt{2})^m$ \citep{West1995247}.
\end{proof}

\section{Counting results and the Impartial Culture assumption}
\label{sec:IC}

As in the previous section, let $a(n,m, \GammaSP)$ denote the number of  single-peaked elections.
In this section, we prove a lower and upper bound on $a(n,m, \GammaSP)$.
These two bounds are asymptotically optimal, i.e., the lower bound converges to the upper bound for every fixed $m$ and $n\rightarrow \infty$.
In addition, we prove exact enumeration results for $a(2,m, \GammaSP)$, $a(n,3, \GammaSP)$ and $a(n,4, \GammaSP)$.

Our results immediately imply bounds on the probability that an $(n,m)$-election is single-peaked assuming that elections are drawn uniformly at random, i.e., according to the Impartial Culture assumption.
The probability is simply $a(n,m, \GammaSP)/(m!)^n$.

\begin{theorem}
It holds that 
\begin{align*}
\frac{m!}{2}\cdot 2^{(m-1)\cdot n} \cdot (1-\epsilon(n,m))\leq a(n,m, \GammaSP)\leq \frac{m!}{2}\cdot 2^{(m-1)\cdot n},
\end{align*}
 where $\epsilon(n,m) \rightarrow 0$ for every fixed $m$ and $n \rightarrow \infty$.
\label{thm:count-sp}
\end{theorem}

\begin{proof}
First observe that an election is single-peaked with respect to an axis if and only if it is single-peaked with respect to its reverse, i.e., the axis read from right to left. 
Thus the total number of axes on $m$ candidates that need to be considered is $m!/2$.
Second, recall that the number of votes that are single-peaked with respect to a given axis is $2^{m-1}$ (cf.~Remark~\ref{rem:number-of-sp-votes}).

Now we have gathered all facts necessary for the upper bound. 
For every one of the $m!/2$ axes considered, select an ordered set of votes from the $2^{m-1}$ votes that are single-peaked with respect to this axis. 
There are exactly $2^{(m-1)\cdot n}$ such possibilities, which yields the upper bound.
Since an election may be single-peaked with respect to more than two axes, this number is only an upper bound for $a(n,m, \GammaSP)$.

Let us turn to the lower bound.
Given a vote $V$, there are only two axes with respect to which both $V$ and its reverse $\bar{V}$ are single-peaked, namely the total orders $V$ and $\bar{V}$ themselves.
Thus the presence of the votes $V$ and $\bar{V}$ in an election forces the axis to be equal to either $V$ or $\bar{V}$.
If we fix a vote $V$, the number of single-peaked elections containing both $V$ and $\bar{V}$ can be determined easily using the inclusion-exclusion principle. 
Indeed, since $2^{(m-1)\cdot n}$ is the number of single-peaked elections for the axis $V$, $\left( 2^{m-1}-1\right)^n$ is the number of single-peaked elections for the axis $V$ that do not contain the vote $V$ (analogous for single-peaked elections that do not contain $\bar{V}$) and $\left( 2^{m-1}-2\right)^n$ is the number of single-peaked elections for the axis $V$ that neither contain $V$ nor $\bar{V}$.
Thus, the number of single-peaked elections containing both $V$ and $\bar{V}$ for some fixed vote $V$ is equal to:
\begin{equation*}
\ 2^{(m-1)\cdot n}  -2\cdot \left( 2^{m-1}-1\right)^n + \left( 2^{m-1}-2\right) ^n.
\end{equation*}
Multiplying this by the number of possibilities for the vote $V$ leads to the lower bound for $a(n,m,\Gamma_{sp})$:
\begin{align*}
 &\ \frac{m!}{2} \cdot \left( 2^{(m-1)\cdot n} + \left( 2^{m-1}-2\right) ^n -2\cdot \left( 2^{m-1}-1\right)^n \right) \\
= & \ \frac{m!}{2} \cdot 2^{(m-1)\cdot n} \cdot \left( 1- \epsilon(n,m) \right), \\
&\text{\ \ \ where }\epsilon(n,m) = \ \frac{2 \cdot(2^{m-1}-1)^n - \left( 2^{m-1}-2 \right)^n}{2^{(m-1)n}}. 
\end{align*}
Since $\epsilon(n,m)\leq 2\cdot \left(\frac{2^{m-1}-1}{2^{m-1}}\right)^n$, $\epsilon(n,m)$ tends to $0$ for every fixed $m$ and $n \rightarrow \infty$.
Clearly, not all single-peaked elections contain a pair of votes where one is the reverse of the other.
Thus this number is indeed only a lower bound.
\end{proof}

In the next theorem we prove exact enumeration formul\ae\ for $a(n,m, \GammaSP)$ for $n=2$, $m=3$ and $m=4$.
Note that for $m\leq 2$ and for $n=1$ all $(n,m)$-elections are single-peaked.
For $n>2$ and for $m>4$ we have not been able to find exact enumeration formulas.

\begin{theorem}
It holds that 
\begin{enumerate}[label=(\roman*.),labelwidth=\widthof{(iii)},itemindent=1em]
\item $a(2,m, \GammaSP)=m!\cdot\binom{2m-2}{m-1}$ for $m \geq 1$,
\item $a(n,3, \GammaSP)=6\cdot 2^{n-1}\left(2^n-1\right)$ and
\item $a(n,4, \GammaSP)=24\cdot 4^{n-1}\cdot\left(2^{n+1}-3\right)$.
\end{enumerate}%
\label{thm:SP_small_profiles}%
\end{theorem}
\begin{proof}
(i.) $a(2,m, \GammaSP)=m!\cdot\binom{2m-2}{m-1}$:
This follows from Lemma~\ref{lem:enumeration_confavoiding<->pp}. We choose the first vote arbitrarily ($m!$ possibilities).
The second vote has to be chosen in such a way that all configurations that characterize single-peakedness are avoided. Since we consider only elections with two votes, the relevant configurations are $(dabc,dcba)$,
$(adbc,dcba)$,
$(dabc,cdba)$ and  
$(adbc,cdba)$ (Theorem~\ref{thm:sp-confs}).
We obtain the permutations $\pi_1=p(dabc,dcba)=1432$, $\pi_2=p(adbc,dcba)=4132$, $\pi_3=p(dabc,cdba)=2431$ and $\pi_4=p(adbc,cdba)=4231$.
Their inverses are $\pi_1^{-1}=\pi_1$, $\pi_2^{-1}=\pi_3$, $\pi_3^{-1}=\pi_2$ and $\pi_4^{-1}=\pi_4$.
Thus, the number of $a(2,m, \GammaSP)=S_m(\pi_1, \pi_2, \pi_3, \pi_4)$, which is equal to $\binom{2m-2}{m-1}$, as shown by \cite{guibert1995permutations} and \cite{bruner2015patterns}.

(ii.) $a(n,3, \GammaSP)=m!\cdot 2^{n-1}\left(2^n-1\right)$: 
We consider all elections with three candidates.
There are $m!$ many possibilities for the first vote $V_1$.
Without loss of generality, let us consider only the vote $V_1: abc$.
Since we have only three candidates, single-peakedness boils down to having at most two last ranked candidates (cf.~Theorem~\ref{thm:bal-hae-sp}).
Due to our assumption that $V_1: abc$, we distinguish three cases: elections in which the votes rank either $a$ or $c$ last, elections in which the votes rank either $b$ or $c$ last and elections in which all votes rank $c$ last.
The number of elections in which the votes rank either $a$ or $c$ last can be determined as follows:
every vote can either be $abc$, $bac$, $cba$ or $bca$.
Hence, there are $4^{n-1}$ possibilities for elections in which the votes rank either $a$ or $c$ last and where $V_1:abc$ holds.
By the same argument, the number of elections in which the votes rank either $b$ or $c$ last is $4^{n-1}$ as well.
The number of elections where $c$ is always ranked last is $2^{n-1}$.
We obtain a total number of single-peaked elections with a fixed first vote of $4^{n-1}+4^{n-1}-2^{n-1}=2^{n-1}\cdot(2\cdot 2^{n-1}-1)$.
Given that $6$ options for the first vote exist, we obtain the stated enumeration result.

(iii.) $a(n,4, \GammaSP)=24\cdot 4^{n-1}\cdot\left(2^{n+1}-3\right)$:
As in the previous proof, we fix $V_1:abcd$.
This vote $V_1$ already rules out some possible axes.
Indeed, only eight axes are single-peaked axes for $V_1$, namely $A_1:abcd$, $A_2:bacd$, $A_3:cabd$, $A_4:cbad$, and their reverses.
Since the reverse of an axis permits the same single-peaked votes, we have to consider only $A_1, A_2, A_3, A_4$.
For $1 \leq i \leq 4$, let $W_i$ denote the set of four-candidate votes that are single-peaked with respect to axis $A_i$.
We count the number of single-peaked elections with four candidates by using the inclusion-exclusion principle, i.e.,
{\allowdisplaybreaks
\begin{align*}
a(n,4, \GammaSP) =& m!\cdot(\card{W_1}^{n-1}+\card{W_2}^{n-1}+\card{W_3}^{n-1}+\card{W_4}^{n-1}\\
&-\card{W_1\cap W_2}^{n-1}-\card{W_1\cap W_3}^{n-1}-\card{W_1\cap W_4}^{n-1}-
\\&-\card{W_2\cap W_3}^{n-1}-\card{W_2\cap W_4}^{n-1}-\card{W_3\cap W_4}^{n-1}
\\&+\card{W_1\cap W_2\cap W_3}^{n-1}+\card{W_1\cap W_2\cap W_4}^{n-1}
\\&+\card{W_1\cap W_3\cap W_4}^{n-1}+\card{W_2\cap W_3\cap W_4}^{n-1}
\\&-\card{W_1\cap W_2\cap W_3\cap W_4}^{n-1})
\end{align*}
}

It is easy to verify that $W_1\cap W_2=W_1\cap W_3 = W_2\cap W_4=\{abcd, bacd\}$.
The remaining intersections look as follows:
\begin{align*}
W_1\cap W_4 &=\{abcd, bacd, bcad, cbad\},\\
W_2\cap W_3 &=\{abcd, bacd, acbd, cabd\},\\
W_3\cap W_4 &=\{abcd, bacd, abdc, badc\}.
\end{align*}
It follows that all intersections of three or four W-sets consist also of $\{abcd, bacd\}$.
The number of votes single-peaked with respect to one axis is $2^{m-1}$ (see Remark~\ref{rem:number-of-sp-votes}), i.e., in our case 8.
We obtain 
\pushQED{\qed} 
\begin{align*}
a(n,4, \GammaSP) &= 4!\cdot(4\cdot 8^{n-1} - 3\cdot 2^{n-1} - 3\cdot 4^{n-1} +  4\cdot 2^{n-1} -  2^{n-1})=\\
&= 24\cdot (4\cdot 8^{n-1} - 3\cdot 4^{n-1}).\qedhere
\end{align*}%
\popQED
\renewcommand{\qedsymbol}{}
\end{proof}

\section{The Impartial Anonymous Culture assumption}
\label{sec:iac}

The counting results from the previous section on the Impartial Culture (IC) assumption can easily be adapted to the Impartial Anonymous Culture (IAC) assumption as we will see in the following.
For the proofs of these results, it is important to keep in mind that an election sampled according to the IAC model is a multiset of votes, i.e., the order of the votes is of no relevance.
Thus, the total number of $(n,m)$-elections is equal to $\multiset{m!}{n}=\binom{m!+n-1}{n}$.

In the following, let $p_{A}(n,m)$ denote the probability that an $(n,m)$-election created according to the IAC assumption is single-peaked.

\begin{proposition}
It holds that 
\begin{enumerate}[label=(\roman*.),labelwidth=\widthof{(iii)},itemindent=1em]
\item \[\frac{m!}{2} \frac{\multiset{2^{m-1}}{n}}{\multiset{m!}{n}} \cdot \left(1 + \epsilon(n,m)\right) \leq p_{A}(n,m) \leq \frac{m!}{2} \frac{\multiset{2^{m-1}}{n}}{\multiset{m!}{n}},\] where $\epsilon(n,m) \rightarrow 0$ as $n \rightarrow \infty$ for $n,m \geq 2,$  
\item $p_{A}(2,m)=\frac{1}{m!+1}\left( \binom{2m-2}{m-1} + 1\right)$  for $m \geq 1$ and
\item $p_{A}(n,3)=\frac{60n}{(n+2)(n+3)(n+4)}$ for $n\geq 1$.
\end{enumerate}%
\label{thm:iac}
\end{proposition}

\begin{proof}
We follow the proofs of Theorems~\ref{thm:count-sp} and \ref{thm:SP_small_profiles}.
\begin{enumerate}[label=(\roman*.),itemindent=1em]
\item Given a fixed axis, there are $\multiset{2^{m-1}}{n}$ $(n,m)$-elections that are single-peaked with respect to this axis.
Multiplying with the number of axes that need to be considered and dividing by the total number of elections leads to the upper bound on the probability.

For the lower bound, we fix a vote $V$ and determine the number of elections that are single-peaked and contain both $V$ and $\bar{V}$:
\begin{align*}
& \multiset{2^{m-1}}{n}-2\cdot \multiset{2^{m-1}-1}{n} + \multiset{2^{m-1}-2}{n} \\
= & \multiset{2^{m-1}}{n} \cdot \left( 1 - 2 \cdot \frac{2^{m-1} -1}{2^{m-1}+n-1}+ \frac{(2^{m-1}-1)(2^{m-1}-2)}{(2^{m-1}+n-1)(2^{m-1}+n-2)}\right) \\
= & \multiset{2^{m-1}}{n} \cdot \left( 1 + \epsilon(n,m)\right), 
\end{align*}
where it can be checked easily that $\epsilon(n,m) \rightarrow 0$ as $n$ tends to infinity.
This gives the lower bound.
\item We pick one vote at random, the second vote can be one of $\binom{2m-2}{m-1}$ possibilities (cf.\ Theorem~\ref{thm:SP_small_profiles}).
However, the order of votes does not matter and we are thus double-counting profiles that consist of two distinct votes--there are $m!(\binom{2m-2}{m-1}-1)/2$ such profiles.
Adding the $m!$ profiles in which the same vote occurs twice, we obtain the following number of possibilities:
\[
\frac{m!}{2}\left( \binom{2m-2}{m-1} + 1\right) 
\]
Dividing by the total number of $(2,m)$-elections leads to the probability.
\item An election with three candidates is single-peaked if and only if it has at most two last-ranked candidates.
Using inclusion-exclusion one obtains that the total number of possible elections is:
\[
3\cdot \left( \multiset{4}{n}-\multiset{2}{n}\right) = \frac{n+1}{2} \left( (n+2)(n+3)-6\right).
\]
Again, dividing by the total number of elections gives the probability.
\end{enumerate}
\end{proof}

The case with 4 candidates that corresponds to case \textit{(iv)} in Theorem~\ref{thm:SP_small_profiles} can not be directly derived from the IC case and would need a far more involved inclusion-exclusion argument. 
It is thus omitted here.

\section{The P\'{o}lya urn model}
\label{sec:polya}

The P\'{o}lya urn model (also refereed to as the P\'{o}lya-Eggenberger urn model) \citep{urnmodels,berg1985urn,mahmoud2008polya} is an approach to sample elections with a variable degree of \textit{social homogeneity}, i.e., where preferences are not independent but voters tend to have the same preferences as other voters.
In the following the parameter $a$, a non-negative integer, describes the degree of social homogeneity. 
As we will see in a moment, the case $a=0$ corresponds to the Impartial Culture assumption, i.e., a population with no homogeneity.

The setting of the P\'{o}lya urn model for an election with $n$ votes and $m$ candidates can be described as follows.
Consider a large urn containing $m!$ balls.
Every ball represents one of the $m!$ possible votes on the candidate set and has a different color. 
An election is then created by subsequently pulling $n$ balls out of the urn according to the following rule.
The first ball is pulled at random and constitutes the first vote of the election.
Then the pulled ball is returned to the urn and $a$ other balls of the same color are added to the urn.
This procedure is repeated $n$ times until an election consisting of $n$ votes is created.

At a first glance, it might seem that the probability assigned to a certain election within the P\'{o}lya urn model depends on the order of the votes.
However this is not the case: Any election that can be obtained by rearranging a given election $\cp$, i.e, by changing the order of the votes,  has exactly the same probability of occurring as the election $\cp$ itself.
First, when the $i$-th ball is drawn from the urn, i.e., when the $i$-th vote is chosen, there are always $m! + (i-1)\cdot a$ balls present in the urn.
Second, for any vote $V$ the number of balls corresponding to $V$, i.e., the number of favourable cases,  only depends on how often the vote $V$ has already been pulled out of the urn and is equal to $(1+k\cdot a)$ where $k$ is the number of times $V$ has already been pulled.

It is now easy to give a concise characterization of this discrete distribution.
In order to alleviate notation, let us use the so called \textit{Pochhammer $k$-symbol} as introduced by \citet{diaz2007hypergeometric}.
\begin{definition}
The \textit{Pochhammer $k$-symbol} is defined as
 $ (x)_{n,k} = \prod_{i=1}^{n} (x+ (i-1)\cdot k)$ 
where in our context $x\in\mathbb{R}$ and $n, k$ are non-negative integers.
Note that $ (x)_{n,1}=x(x+1)(x+2)\ldots (x+ n-1)$ is the ordinary Pochhammer symbol (also known as rising factorial) and $ (1)_{n,1}=n!$.
\end{definition}
We can now define the probability of a given $(n,m)$-election with $\ell$ distinct votes.
Let $n_i, i \in [m!]$ be non-negative integers with $\sum_{i=1}^{m!} n_i=n$ such that, for all $i\in[\ell]$ vote $V_i$ appears $n_i$ times.
The probability of such an election is given by:
\begin{align}
\binom{n}{n_1, \ldots, n_{\ell}} \cdot \frac{\prod_{i=1}^{\ell} (1)_{n_i,a}}{(m!)_{n,a}} = \frac{n!}{\prod_{i=1}^{\ell}n_i!} \cdot \frac{\prod_{i=1}^{\ell} (1)_{n_i,a}}{(m!)_{n,a}}.
\label{eqn:distribution_Polya}
\end{align}

Note that setting $a=0$ corresponds to the case where every one of the $n$ votes is drawn from exactly the same urn, namely the urn containing every one of the $m!$ balls exactly once.
Thus, the votes are chosen independently and every vote has the same probability of occurring; this corresponds to the Impartial Culture assumption.

Setting the homogeneity factor to $a=1$ leads to the Impartial Anonymous Culture (IAC) assumption in which elections are considered as multisets  and not as lists of votes and every election has the same probability of occurring.
This can be seen as follows: Under IAC an election is fully characterized by the numbers $n_i$, $i \in [m!]$, as defined above.
Setting $a=1$ in equation~\eqref{eqn:distribution_Polya} in order to determine the probability of an election in which vote $V_i$ appears $n_i$ times, we obtain:
\[
\frac{n!}{\prod_{i=1}^{\ell}n_i!} \cdot \frac{\prod_{i=1}^{\ell} (1)_{n_i,1}}{(m!)_{n,1}}=
\frac{n!}{\prod_{i=1}^{\ell}n_i!} \cdot \frac{\prod_{i=1}^{\ell} n_i!}{\prod_{j=1}^{n} (m!+ (j-1))}=
\frac{1}{\binom{m!+n-1}{n}}.
\]
Clearly, this probability does not depend on the choice of the numbers $n_i$, $i \in [m!]$, and thus every multiset of votes has the same probability of being sampled.

Since we have consider IC and IAC already in Section~\ref{sec:IC} and \ref{sec:iac} and have obtained asymptotically optimal results, the following lower bound theorem is interesting only for $a>1$.

\begin{theorem}
Let $p_P(n,m,a)$ denote the probability that an $(n,m)$-election created according to the P\'{o}lya urn model with homogeneity $a>0$ is single-peaked.
It holds that:
\begin{align*}
p_P(n,m, a) %
\geq  \frac{m! (n-1)!}{a \left( \frac{m!}{a}\right)_{n,1}} \cdot \Bigg[&1+ \frac{2}{a} \binom{2m-2}{m-1}H_{n-1} + \\
& +\frac{n}{a} \sum_{l=2}^{n-1} \frac{(2^{m-1}-2)_{n-l,a}}{a^{n-l}} \cdot \frac{H_{l-1}}{l(n-l)!}  \Bigg], 
\end{align*}
where $H_k$ denotes the $k$-th harmonic number $\sum_{i=1}^k \frac{1}{i}$.
\label{thm:polya-bound}
\end{theorem}

\begin{proof}%
Before we start with the actual proof, let us collect a few useful observations.
In the following we will use that
\begin{align}
(m!)_{n,a} = a^n\cdot\prod_{i=1}^{n}\left(\frac{m!}{a}+(i-1)\right) = a^n \left( \frac{m!}{a}\right)_{n,1}.
\label{eqn:equ_pochhammer_m!}
\end{align}
Moreover, we will use the following bound:
\begin{align}
(1)_{k,a}  = \prod_{i=1}^{k}(1+(i-1)\cdot a) \geq \prod_{i=2}^{k}(i-1)\cdot a = (k-1)! \cdot a^{k-1}.
\label{eqn:bound_pochhammer}
\end{align}
Since it holds that
\[
\frac{1}{i\cdot (l-i)}= \frac{1}{l} \cdot \left( \frac{l-i+i}{i \cdot (l-i)}\right)=  \frac{1}{l} \cdot \left( \frac{1}{ i} + \frac{1}{ l-i}\right) 
\]
the following sum can be expressed with the help of the harmonic numbers:
\begin{align}
\sum_{i=1}^{l-1} \frac{1}{i\cdot (l-i)} = \frac{1}{l} \cdot \left(\sum_{i=1}^{l-1} \frac{1}{i} + \sum_{i=1}^{l-1} \frac{1}{l-i} \right) 
= \frac{2}{l} \cdot \sum_{i=1}^{l-1} \frac{1}{i} 
= \frac{2}{l}H_{l-1}.
\label{eqn:harmonic_numbers}
\end{align}

The proof of the theorem is now split in three parts: First, we consider elections with only one distinct vote. Then, we determine a lower bound on the probability of single-peaked elections that consist of exactly two distinct votes.
Third, we give a lower bound on the number of single-peaked elections that contain at least three distinct votes.

Let us now start with the first part of the lower bound.
Clearly, an election in which all votes are identical is single-peaked.
Let us denote the probability of this event by $p_1$; in the following we fix $m,n,a$ an omit them in the notation.
According to the discrete probability distribution of the P\'{o}lya urn model \eqref{eqn:distribution_Polya}, the probability of this event is:
\begin{align*}
p_1 %
& = m! \cdot \frac{(1)_{n,a}}{(m!)_{n,a}} 
 \geq \frac{m! (n-1)! a^{n-1}}{a^n \left( \frac{m!}{a}\right)_{n,1}}  =  \frac{m!(n-1)!}{a \left( \frac{m!}{a}\right)_{n,1}},
\end{align*}
where we used Equations \eqref{eqn:equ_pochhammer_m!} and \eqref{eqn:bound_pochhammer}.

Next, we want to determine the probability that an election is sampled that is single-peaked and consists of exactly two distinct votes.
Let us denote the probability of this event by $p_2$.
From the first statement of Theorem~\ref{thm:SP_small_profiles} we know that there are exactly $m! \cdot \binom{2m-2}{m-1}$ elections with two voters and $m$ candidates that are single-peaked.
That is, if we pick a first vote $V_1$ at random, there are $\binom{2m-2}{m-1}$ votes $V_2$ that form a single-peaked election together with $V_1$.
We thus have the following (again according to \eqref{eqn:distribution_Polya}):
\begin{align*}
p_2 & = m!  \binom{2m-2}{m-1} \cdot \sum_{i=1}^{n-1} \underbrace{\mathbb{P}(i \text{ votes equal to } V_1 \text{ and } n-i \text{ votes equal to } V_2}_{= p'(i)}) 
\end{align*}
According to \eqref{eqn:distribution_Polya} the probability $p'(i)$ is equal to
\[
p'(i) = \binom{n}{i} \frac{(1)_{i,a} \cdot (1)_{n-i,a}}{(m!)_{n,a}}.
\]
Using the bound in Equation~\eqref{eqn:bound_pochhammer} and the equality in Equation~\eqref{eqn:equ_pochhammer_m!},
$p'(i)$ can be bounded from below as follows: 
\begin{align*}
p'(i) & \geq \binom{n}{i} \frac{(i-1)!\cdot  a^{i-1} \cdot (n-i-1)!\cdot a^{n-i-1}}{(m!)_{n,a}} 
  =  \frac{n!\cdot a^{n-2}  }{a^n \left( \frac{m!}{a}\right)_{n,1} \cdot i(n-i)}.
\end{align*}
For $p_2$ we thus obtain
\begin{align*}
p_2 & \geq m! \cdot \binom{2m-2}{m-1} \frac{n!  }{a^2 \left( \frac{m!}{a}\right)_{n,1}} \sum_{i=1}^{n-1} \frac{1}{  i(n-i)} \\
& = \frac{m! (n-1)!}{\left( \frac{m!}{a}\right)_{n,1}} \cdot \binom{2m-2}{m-1}  \frac{2}{a^2} H_{n-1},
\end{align*}
where the transformation from the first to the second line is done with the identity in Equation~\eqref{eqn:harmonic_numbers}.

Finally, for single-peaked elections that have more than two distinct votes, we only consider elections that contain a vote $V$ and also its reverse $\bar{V}$.
Let us denote the probability of this event by $p_3$.
As in the proof of the bounds under the IC assumption this idea is based on the following fact about single-peakedness: If a vote $V$ and its reverse vote $\bar{V}$ are both present within an election, then there are at most two axes with respect to which this election can be single-peaked, namely the axes $V$ and $\bar{V}$.
Thus, if a single-peaked election  contains both the vote $V$ and $\bar{V}$, all the other votes must be among the $2^{m-1}-2$ votes that are also single-peaked with respect to the axis $V$ (respectively $\bar{V}$).
Let us denote this set of votes that are not equal to $V$ or $\bar{V}$ and that are single-peaked with respect to the axis $V$ by $S_V$.

The probability $p_3$ is then given as follows:
\begin{align*}
p_3 = 
& \frac{m!}{2} \sum_{l=2}^{n-1}\sum_{i=1}^{l-1} \mathbb{P}(
i \text{ votes equal to } V, l-i \text{ votes equal to } \bar{V},   n-l \text{ votes in } S_V) \\
= & \frac{m!}{2} \sum_{l=2}^{n-1}\sum_{i=1}^{l-1}\binom{n}{i}\binom{n-i}{l-i} \frac{(1)_{i,a} \cdot (1)_{l-i,a} \cdot (2^{m-1}-2)_{n-l,a}}{(m!)_{n,a}},
\end{align*}
where $m!/2$ stands for the number of possible choices for the vote $V$, $i$ is the number of times the vote $V$ appears and $l-i$ is the number of times  the vote $\bar{V}$ appears.

By using the bound in Equation~\eqref{eqn:bound_pochhammer} as well as the identity in Equation~\eqref{eqn:harmonic_numbers}, we obtain the following:

\begin{align*}
p_3 & \geq 
 \frac{m!}{2(m!)_{n,a}} \sum_{l=2}^{n-1} (2^{m-1}-2)_{n-l,a} \sum_{i=1}^{l-1}\binom{n}{i}\binom{n-i}{l-i} (i-1)! (l-i-1)! a^{l-2}   \\
& = \frac{m!}{2(m!)_{n,a}} \sum_{l=2}^{n-1} (2^{m-1}-2)_{n-l,a} \cdot a^{l-2} \sum_{i=1}^{l-1} \frac{n!(n-i)!(i-1)!(l-i-1)!}{i!(n-i)!(l-i)!(n-l)!} \\
& = \frac{m!n!}{2(m!)_{n,a}} \sum_{l=2}^{n-1} (2^{m-1}-2)_{n-l,a} \cdot \frac{a^{l-2}}{(n-l)!} \sum_{i=1}^{l-1} \frac{1}{i(l-i)} \\
& = \frac{m!n!}{\left( \frac{m!}{a}\right)_{n,1}} \sum_{l=2}^{n-1} \frac{(2^{m-1}-2)_{n-l,a}}{a^{n+2-l}} \cdot \frac{1}{l(n-l)!} H_{l-1}.
\end{align*}

Since $p_1 + p_2 + p_3\leq p_P(n,m,a)$, we obtain the desired lower bound.
\end{proof}

To illustrate the rather involved lower bound of Theorem~\ref{thm:polya-bound}, we consider the special case of $a=m!$. This special case corresponds to highly homogeneous elections; the probability that the first and the second vote are identical is roughly 50\%.
It is a typical assumption that $a$ is a multiple of $m!$
\citep{slinko06,ecai/Walsh10,jair/Walsh11} since otherwise, i.e., for a fixed $a$, the actual homogeneity of elections drawn according to the P\'{o}lya urn model would depend on the number of candidates $m$.

\begin{corollary}
Let $p_P(n,m,m!)$ denote the probability that an $(n,m)$-election created according to the P\'{o}lya urn model with homogeneity $m!$ is single-peaked.
It holds that:
\[p_P(n,m, m!)\geq \frac{1}{n}\cdot \left(1 + 2\frac{\ln(n-1)}{m!} \cdot  \frac{(2m-2)!}{((m-1)!)^2}\right).\]
\label{cor:polya}%
\end{corollary}

We see that for $a=m!$ and small $n$, there is a significant probability that the P\'{o}lya urn model produces single-peaked elections.

\section{Mallows model}
\label{sec:mallow}

The Mallows model \citep{mallows} assumes that there is a \emph{reference vote} and votes are more likely to appear in an election if they are close to this reference vote. Closeness is measured by the Kendall tau rank distance, defined as follows.

\begin{definition}
Given two votes $V$ and $W$ contained in an election $\cp$, the \emph{Kendall tau rank distance} $\kappa(V,W)$ is a metric that counts the number of pairwise disagreements between $V$ and $W$. 
To be more precise:
\[\kappa(V,W) = \card{\{\{c,c'\}\subseteq C: ( V:cc' \wedge W:c'c ) \vee ( V:c'c \wedge W:cc' )\}}. 
\]
\end{definition}

Note that $\kappa(V,W)$ is also the minimum number of transpositions, i.e., swaps, of adjacent elements, needed to transform $V$ into $W$ or vice versa.
We can now define the Mallows model.

\begin{definition}
Let $C$ be a set of candidates with $\card{C}=m$ and let $T(C)$ be the set of all total orders on $C$.
Given a \emph{reference vote} $V$ and a real number $\phi \in (0,1]$, the so-called \emph{dispersion parameter}, the \emph{Mallows model} is defined as follows.
Every vote $W$ of an $(n,m)$-election is determined independently from the others according to the following probability distribution:
\begin{align}
\mathbb{P}_{V,\phi}(W)= \frac{1}{Z} \cdot \phi^{\kappa(V,W)},
\label{eqn:distribution_mallows}
\end{align}
where the normalization constant $Z=\sum_{W \in T(C)} \phi^{\kappa(V,W)}$ fulfils $Z=1 \cdot (1 + \phi)\cdot (1+ \phi +\phi^2) \cdots (1+  \ldots + \phi^{m-1}) $.
\end{definition}
Note that choosing $\phi=1$ corresponds the Impartial Culture assumption and as $\phi \rightarrow 0$ one obtains a distribution that concentrates all mass on $V$.

\begin{theorem}
Let $p_M(n,m,\phi)$ denote the probability of an $(n,m)$-election being single-peaked if it is created according to the Mallows model with dispersion parameter $\phi$.
Then the following lower bound holds:
\[
p_M(n,m,\phi) \geq \left( \frac{1 + \phi\cdot(m-1) + \phi^2 \cdot (m-2)(m-3)/2}{Z}\right)^n.
\]\label{thm:mallow}
\end{theorem}

\begin{proof}
Without loss of generality, we can assume that the reference vote is $V: c_1 c_2 \ldots$ $c_m$.
We define an axis $A$ as follows: $A: \ldots c_6 c_4 c_2  c_1   c_3   c_ 5  \ldots$.
Clearly, $V$ is single-peaked with respect to $A$ and it will turn out that ``many'' other votes that are close to $V$ with respect to the Kendall tau distance are also single-peaked with respect to this axis.
See Figure~\ref{fig:mallows_axis} for a representation of the axis $A$ and the reference vote $V$ for the case of seven candidates.
In the following, we will write ``$W$ is SP'' as a short form of  ``the total order $W$ on $C$ is single-peaked with respect to axis $A$''.

\begin{figure}
\centering
\begin{tikzpicture}[yscale=0.65,xscale=0.9]

  \def\xmin{1}
  \def\xmax{7}
  \def\ymin{0}
  \def\ymax{7}
  
  \draw[step=1cm,black!20,very thin] (\xmin,\ymin) grid (\xmax,\ymax);

  \draw[->] (\xmin -0.5,\ymin) -- (\xmax+0.5,\ymin) node[right] {axis $A$};
  \foreach \x/\xtext in {1/c_7, 2/c_5, 3/c_3, 4/c_1, 5/c_2, 6/c_4, 7/c_6}
    \draw[shift={(\x,\ymin)}] (0pt,2pt) -- (0pt,-2pt) node[below] {$\xtext$};

  \foreach \x/\y in {4/7,5/6,3/5, 6/4, 2/3, 7/2, 1/1}
    \node[fill=black, circle, inner sep=0.6mm] at (\x,\y) {};
    
  \draw[solid,black] (1,1)--(2,3)--(3,5)--(4,7) -- (5,6)--(6,4)--(7,2);

  \foreach \x/\y in {4/7,5/5,6/6, 3/4, 2/3, 7/2, 1/1}
    \node at (\x,\y) {$\star$};
   
   \draw[dashed,black] (1,1)--(2,3)--(3,4)--(4,7) -- (5,5)--(6,6)--(7,2);
\end{tikzpicture}
\caption{The axis $A$ and the reference vote $V:c_1 c_2 c_3 c_4 c_5 c_6 c_7$, shown as a solid line. The dashed line represents the vote $W:c_1 c_4 c_2 c_3 c_5 c_6 c_7$ and is not single-peaked with respect to $A$. The Kendall tau distance of $V$ and $W$ is 2. Note that all votes with a Kendall tau distance to $V$ of $1$ are single-peaked with respect to $A$.}
\label{fig:mallows_axis}
\end{figure}
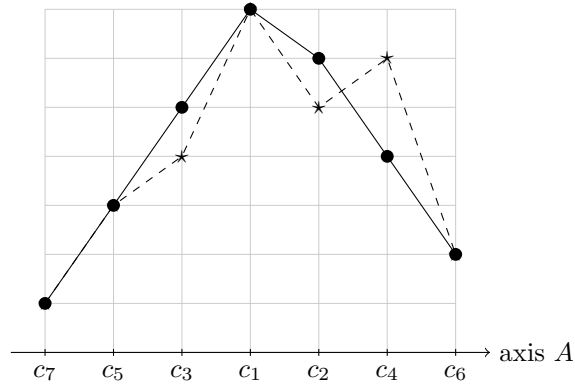

The idea of this proof is to bound the probability $p_M(n,m,\phi)$ from below as follows:
\begin{align*}
p_M(n,m,\phi) \geq & \left( \mathbb{P}_{V,\phi}(W: W \text{ is SP})\right)^n.
\end{align*}
Moreover, we use the following bound:
\begin{align*}
\mathbb{P}_{V,\phi}(W: W \text{ is SP})  \geq\ & \mathbb{P}_{V,\phi}(V) \\
&+\sum_{\stackrel{W \text{ is SP} \text{ and }}{ \kappa(V,W)=1}} \mathbb{P}_{V,\phi}(W)
+\sum_{\stackrel{W \text{ is SP} \text{ and }}{ \kappa(V,W)=2}} \mathbb{P}_{V,\phi}(W)
\end{align*}

First, it is clear that $\mathbb{P}(V)=1/Z$.

Second, we need to compute the number of votes $W$ that are single-peaked with respect to $A$ and that fulfil $\kappa(V,W)=1$.
Votes $W$ with $\kappa(V,W)=1$ are votes in which the order of exactly one pair of candidates $(c_i, c_{i+1})$ has been changed in $V$. 
Since there are $(m-1)$ pairs of adjacent candidates in $V$, there are exactly $(m-1)$ votes $W$ with $\kappa(V,W)=1$.
All these votes are single-peaked with respect to the axis $A$ since:
\begin{itemize}
\item If $c_1$ and $c_2$ are interchanged, the position of the peak on the axis $A$ is changed, but clearly no new peaks arise.
\item If two other candidates $c_i$ and $c_{i+1}$ are interchanged, one of these two candidates lies to the left of the peak on $A$ and the other one to the right of the peak. 
Thus, interchanging only these two candidates does not create a new peak either.
\end{itemize}
Therefore we have the following:
\[
\sum_{\stackrel{W \text{ is SP} \text{ and }}{ \kappa(V,W)=1}} \mathbb{P}_{V,\phi}(W) = (m-1) \cdot \frac{\phi}{Z}.
\]

Third, we need to compute the number of votes $W$ that are single-peaked with respect to $A$ and that fulfil $\kappa(V,W)=2$. 
Here we have a different situation: Not all votes that can be obtained by exactly two swaps of adjacent candidates in $V$ are single-peaked with respect to $A$.
For instance, first swapping the candidates $(c_3, c_4)$ and then swapping $(c_2, c_4)$ in $V$, does not lead to a vote that is single-peaked with respect to $A$.
For this example, see the vote shown as a dashed line in Figure~\ref{fig:mallows_axis}.
The problem here is that the swapping of these two pairs changes the order of $c_2$ and $c_4$, two elements that both lie on the same side of the peak on $A$, and thus a valley is created by the elements $c_1, c_2$ and $c_4$.
In general, a pair of swaps $(c_i, c_{i+1})$ and $(c_j, c_{j+1})$ is always allowed if the two pairs of candidates do not have any elements in common.
Note that the order of the two (disjoint) swaps is of no importance and without loss of generality we can assume that $i+1 <j$.

Knowing this, we can bound
the number of votes $W$ with $\kappa(V,W)=2$ that are single-peaked with respect to $A$ as follows:
If the first swap is $(c_i, c_{i+1})$ for some $i \in [1,m-3]$ and the second swap $(c_j, c_{j+1})$ is disjoint from the first one, $j$ has to fulfil $j \in [i+2,m-1]$ and thus there are $(m-i-2)$ possibilities for $(c_j, c_{j+1})$.
Summing over all possible $i$ we obtain that there are at least
\[
\sum_{i=1}^{m-3}m-i-2 = \frac{(m-2)(m-3)}{2}
\]
many votes $W$ with $\kappa(V,W)=2$ that are single-peaked with respect to $A$.
Thus we have
\[
\sum_{\stackrel{W \text{ is SP} \text{ and }}{ \kappa(V,W)=2}} \mathbb{P}_{V,\phi}(W) \geq \frac{\phi^2}{Z}\cdot \frac{(m-2)(m-3)}{2}.
\]

Putting the results for Kendall tau distance equal to 0, 1 and 2 together we obtain the desired lower bound.
\end{proof}

The lower bound result of Theorem~\ref{thm:mallow} does not give an immediate intuition for the likelihood of single-peakedness under Mallows model. Hence we consider the special case $\phi={1\over m}$. This substitution yields a simpler lower bound, which is considerably larger than, e.g., the lower bound of roughly $(2^m/m!)^n$ obtained for the Impartial Culture Assumption (Theorem~\ref{thm:count-sp}).
For a short discussion on ``realistic'' parameter values $\phi$ we refer to Section~\ref{sec:num}.

\begin{corollary}
Assuming $\phi={1\over m}$, it holds that
\begin{align*}
p_M(n,m,\phi)\geq \left( 1.5 \left( \frac{1 - \frac{1}{m}}{1- \left( \frac{1}{m}\right)^m}\right)^{m-1} \right)^n > \left(1-{{1\over m}}\right)^{(m-1)n}.
\end{align*}
\label{cor:mallows}%
\end{corollary}

\begin{proof}
Inserting $\phi={1\over m}$ in the lower bound of Theorem~\ref{thm:mallow} yields
\begin{equation}
p_M(n,m,\phi) \geq \left( \frac{1 + \frac{m-1}{m} + \frac{(m-2)(m-3)}{2m^2}}{\left(1+{1\over m}\right)\cdot \left(1+{1\over m}+{1\over m^2}\right) \cdots \left(1+{1\over m}+\dots+{1\over m^{m-1}}\right)}\right)^n.
\label{eqn:mallow_1/m_bound}
\end{equation}
As can be checked easily, the numerator is larger than $1.5$.
Moreover the finite geometric sums in the denominator are all bounded from above by 
\[
\frac{1- \left( \frac{1}{m}\right)^m}{1 - \frac{1}{m}}.
\]
We thus obtain
\[
p_M(n,m,\phi) \geq \left( 1.5 \left( \frac{1 - \frac{1}{m}}{1- \left( \frac{1}{m}\right)^m}\right)^{m-1} \right)^n.
\]
Replacing the numerator in \eqref{eqn:mallow_1/m_bound} by $1$ and every one of the finite geometric sums in the denominator by an infinite geometric row leads to the second, much rougher lower bound. 
Note that this is simply a bound on the probability of sampling an $(n,m)$-election in which all votes are equal to the reference vote.
\end{proof}

\section{Numerical Evaluations}
\label{sec:num}

\begin{figure}
\begin{center}
{
\renewcommand{\tabcolsep}{0.16cm}
\begin{tabular}{r@{\hskip 10pt}c}
\toprule $(n,m)$ & exact probability \\
\midrule
 $(2,3)$ & 1   \\
 $(5,3)$ &  0.38  \\
 $(10,3)$ &  0.05  \\
 $(25,3)$ & $1.19\cdot 10^{-4}$     \\
 $(50,3)$ & $4.70\cdot 10^{-9}$     \\\midrule 
 $(2,4)$ &  0.83 \\
 $(5,4)$ &  0.05 \\
 $(10,4)$ & $2.03\cdot 10^{-4}$   \\
 $(25,4)$ & $1.42\cdot 10^{-11}$  \\
 $(50,4)$ & $1.67\cdot 10^{-23}$   \\
 \bottomrule
\end{tabular}
\qquad
\begin{tabular}{r@{\hskip 10pt}c@{\hskip 10pt}c}
\toprule $(n,m)$ & lower bound & upper bound \\
\midrule 
 $(2,5)$ & \multicolumn{2}{c}{0.58}\\
 $(5,5)$ & $1.6 \cdot 10^{-4}$ & $2.6 \cdot 10^{-3}$  \\
 $(10,5)$ & $2.2 \cdot 10 ^{-8}$ & $1.1 \cdot 10^{-7}$  \\
 $(25,5)$ & $5.0 \cdot 10 ^{-21}$ & $8.0 \cdot 10^{-21}$  \\
 $(50,5)$ & $9.7 \cdot 10^{-43}$ & $1.1 \cdot 10^{-42}$  \\\midrule
 $(2,10)$ & \multicolumn{2}{c}{$1.3 \cdot 10^{-2}$}\\
 $(5,10)$ & $7.6 \cdot 10^{-18}$ & $1.1 \cdot 10^{-13}$  \\
 $(10,10)$ & $1.9 \cdot 10 ^{-36}$ & $5.7 \cdot 10^{-33}$  \\
 $(25,10)$ & $2.3 \cdot 10 ^{-93}$ & $1.0 \cdot 10^{-90}$  \\ 
 $(50,10)$ & $4.6 \cdot 10 ^{-189}$ & $5.5 \cdot 10^{-187}$ \\
 \bottomrule
\end{tabular}
}
\end{center}
\caption{The likelihood that an $(n,m)$-election is single-peaked when drawn according to the \textbf{Impartial Culture assumption}. The probabilities in the left table are obtained via Theorem~\ref{thm:SP_small_profiles}; those in the right table are obtained via Theorem~\ref{thm:count-sp} (except for $n=2$, which are obtained via Theorem~\ref{thm:SP_small_profiles} and are exact).}
\label{tab:results-ic}
\end{figure}

\begin{figure}
\begin{center}
{
\renewcommand{\tabcolsep}{0.16cm}
\begin{tabular}{r@{\hskip 10pt}c@{\hskip 10pt}c}
\toprule $(n,m)$ &lower bound & upper bound \\
\midrule
 $(2,3)$ &  \multicolumn{2}{c}{1}  \\
 $(5,3)$ &  \multicolumn{2}{c}{0.59}  \\
 $(10,3)$ &  \multicolumn{2}{c}{0.27}  \\
 $(25,3)$ &   \multicolumn{2}{c}{0.068}   \\
 $(50,3)$ &   \multicolumn{2}{c}{0.02} \\\midrule
 $(2,4)$ & \multicolumn{2}{c}{0.84}  \\
 $(5,4)$ & $1.46 \cdot 10^{-2}$  & $9.67 \cdot 10^{-2}$\\
 $(10,4)$ & $8.34 \cdot 10^{-4}$ &  $2.52 \cdot 10^{-3}$\\
 $(25,4)$ & $7.89 \cdot 10^{-7}$ &  $1.30 \cdot 10^{-6}$\\
 $(50,4)$ & $4.28 \cdot 10^{-10}$  & $5.58 \cdot 10^{-10}$\\
 \bottomrule
\end{tabular}
\qquad
\begin{tabular}{r@{\hskip 10pt}c@{\hskip 10pt}c}
\toprule $(n,m)$ & lower bound & upper bound \\
\midrule
 $(2,5)$ & \multicolumn{2}{c}{0.58}\\
 $(5,5)$ & $2.17 \cdot 10^{-4}$ & $4.13 \cdot 10^{-3}$  \\
 $(10,5)$ & $1.19 \cdot 10^{-7}$ &  $7.98 \cdot 10^{-7}$ \\
 $(25,5)$ & $1.44 \cdot 10^{-16}$ & $3.76 \cdot 10^{-16}$ \\
 $(50,5)$ & $2.91 \cdot 10^{-28}$ & $4.94 \cdot 10^{-28}$ \\\midrule
 $(2,10)$ & \multicolumn{2}{c}{$1.3 \cdot 10^{-2}$}\\
 $(5,10)$ & $7.78 \cdot 10^{-18}$ & $1.03 \cdot 10^{-13}$  \\
 $(10,10)$ & $2.06 \cdot 10^{-36}$ & $6.19 \cdot 10^{-33}$  \\
 $(25,10)$ & $3.69 \cdot 10^{-93}$ &  $1.76 \cdot 10^{-90}$ \\ 
 $(50,10)$ & $4.29 \cdot 10^{-188}$ & $5.05 \cdot 10^{-186}$ \\
 \bottomrule
\end{tabular}
}
\end{center}
\caption{The likelihood that an $(n,m)$-election is single-peaked when sampled according to the \textbf{Impartial Anonymous Culture} assumption. The probabilities in both tables are obtained via Proposition~\ref{thm:iac}.}
\label{tab:results-iac}
\end{figure}

\begin{figure}
\begin{center}
{
\renewcommand{\tabcolsep}{0.16cm}
\begin{tabular}{r @{\hskip 9pt}c@{\hskip 9pt}c@{\hskip 9pt}c}
\toprule $(n,m)$  & $a=10$ & $a=m!/2$ & $a=m!$ \\
\midrule
 $(10,5)$ &  $1.6\cdot10^{-4}$ & $0.13$ & $ 0.43$ \\
 $(25,5)$ &  $8.4\cdot10^{-8}$ & $3.0\cdot10^{-2}$ & $ 0.21$\\
 $(50,5)$ &  $1.5\cdot10^{-10}$ & $9.1\cdot10^{-3}$ & $ 0.12$\\
 \midrule
 $(10,10)$ &  $3.6\cdot10^{-36}$ & $2.0\cdot10^{-2}$ & $ 0.10$ \\
 $(25,10)$ &  $2.3\cdot10^{-91}$ & $3.6\cdot10^{-3}$ & $ 4.4\cdot10^{-2}$ \\
 $(50,10)$ & $2.6\cdot10^{-181}$ & $9.7\cdot10^{-4}$ & $ 2.2\cdot10^{-2}$ \\
 \bottomrule
\end{tabular}
}
\end{center}
\caption{Lower bounds  obtained from Theorem~\ref{thm:polya-bound} on the likelihood that an $(n,m)$-election is single-peaked when sampled according to the \textbf{P\'{o}lya urn model} with homogeneity $a$.}
\label{tab:results-polya}
\end{figure}

\begin{figure}
\begin{center}
{
\renewcommand{\tabcolsep}{0.16cm}
\begin{tabular}{r@{\hskip 9pt}c@{\hskip 9pt}c@{\hskip 9pt}c@{\hskip 9pt}c@{\hskip 9pt}c}
\toprule $(n,m)$ & $\phi=0.3$& $\phi=0.2$& $\phi=0.1$ & $\phi=0.05$ & $\phi=0.01$ \\
\midrule
$(10,5)$ &0.02 &$0.15$& $ 0.59$ & $0.86$ & $ 0.99$ \\
 $(25,5)$ & $5.7 \cdot 10^{-5}$ &$8.7\cdot10^{-3}$ & $ 0.26$ & $0.70$ & $0.98$ \\
 $(50,5)$ & $3.3 \cdot 10^{-9}$&$7.6 \cdot 10^{-5}$& $7.2\cdot 10^{-2}$ & $0.49$ & $0.97$ \\
 \midrule
 $(10,10)$ & $3.7 \cdot 10^{-6}$ &$2.7 \cdot 10^{-3}$& $0.20$ & $0.66$ & $0.98$ \\
 $(25,10)$ & $2.7 \cdot 10^{-14}$&$3.7 \cdot 10^{-7}$& $1.9\cdot 10^{-2}$ & $0.36$ & $0.96$ \\
 $(50,10)$ &$7.5 \cdot 10^{-28}$  &$1.4 \cdot 10^{-13}$& $3.7 \cdot 10^{-4}$ & $0.13$  & $0.92$ \\
 \bottomrule
\end{tabular}
}
\end{center}
\caption{Lower bounds via Theorem~\ref{thm:mallow} on the likelihood that an $(n,m)$-election is single-peaked when sampled according to the \textbf{Mallows model} with dispersion parameter $\phi$.}
\label{tab:results-mallows}
\end{figure}

In this section we provide numerical evaluations of our probability results from the previous sections and make some observations based on these evaluations.
In Table~\ref{tab:results-ic}, we list exact probabilities that an $(n,m)$-election is single-peaked assuming the Impartial Culture assumptions for small values of $m$ and bounds for these probabilities for a larger number of candidates.
Table~\ref{tab:results-iac} shows probabilities for elections of the same size assuming the Impartial Anonymous Culture.
Finally, Table~\ref{tab:results-polya} shows lower bounds for the P\'{o}lya urn model and Table~\ref{tab:results-mallows} shows lower bounds for the Mallows model.

We conclude this section with the following observations:
\begin{itemize}
\item The probabilities shown in Table~\ref{tab:results-ic} illustrate how unlikely it is that an election drawn according to IC is single-peaked. Single-peakedness is a strong combinatorial property, so it is not surprising that is is not satisfied by  elections sampled uniformly at random.
However, it is noteworthy that even for very small $n$ and $m$ the probability is small, e.g., for $m=n=5$ it is less than $0.0026$.
Conversely, our results indicate that even for very small real-world single-peaked data sets it is highly unlikely that their single-peakedness is the product of mere chance.
\item As can be seen in Table~\ref{tab:results-iac}, the probability that an election is single-peaked is slightly higher when it is sampled according to the IAC model than when it is sampled according to the IC model.
This can be explained heuristically as follows: Under the IAC model, elections with many coinciding votes have the same likelihood of appearing as elections consisting of many different votes.
However, in the IC  model, elections consisting of many different votes have a higher chanced of being sampled because of the many ways in which the votes can be rearranged.
It is clear that an election where most votes are the same has a higher chance of being single-peaked than an election in which very many different votes appear.
Thus, it is not surprising if the likelihood of single-peakedness is higher under the IAC than under the IC model.
\item For the P\'{o}lya urn model (Table~\ref{tab:results-polya}) we observe significantly higher probabilities.
This is of course due to our chosen parameter values $a$---recall that $a=0$ implies IC and $a=1$ implies IAC.
In particular for $a=m!$ we see that single-peaked profiles arise with considerable likelihood.
The assumption of $a=m!$ is common in the literature \citep{slinko06,ecai/Walsh10,jair/Walsh11} and even values of up to $a=3m!$ have been considered \citep{LepelleyV03}.
As a consequence, we learn from these probabilities that setting $a=m!$ generates extremely homogeneous profiles, even to the extent that they become single-peaked.
\item Even higher probabilities are shown in Table~\ref{tab:results-polya}.
We see that for $\phi=0.05$ single-peakedness is likely to be observed, e.g., with a probability $>0.49$ for $n=50$ and $m=5$.
Clearly, $\phi=0.05$ is a strong assumption and profiles obtained in this way are highly homogeneous; in fact, $\phi=0.05$ implies that all voters share the same preferences except for minor deviations, which, in turn, enables the single-peakedness property to hold.
For ${\phi=0.1}$ we still see a significant chance of single-peakedness, for larger values of $\phi$ the likelihood deteriorates quickly.
The question arises: what are typical values for $\phi$?
\cite{BetzlerBN-experiments14} compute maximum-likelihood estimates of $\phi$ for different real-world data sets. They find values\footnote{Their parameter $\theta$ is related to $\phi$ via the equation $\phi=e^{-\theta}$.} ranging from $0.7$ to almost $1$.
They also generate elections using values for $\phi$ ranging from $0,37$ to almost $1$.
Other publications generate election with $\phi$ in the interval $[0.3,1]$ \citep{aaai/BoutilierLOP14} and $[0.6, 0.9]$ \citep{ijcai/OrenFB13}.
We see that all these parameter values are too large to imply single-peakedness with non-negligible probability.
On the one hand, this implies that values for $\phi$ small enough to generate single-peakedness profiles are generally too restrictive to be found in (published) experiments.
On the other hand, our results allow to argue that the parameter values in the aforementioned papers have been chosen sensibly since the accordingly generated elections contain (at least) enough disagreement as to prevent single-peakedness to arise with significant likelihood.
\end{itemize}

\section{Conclusions and Directions for Future Research}
\label{sec:concl}

We have seen that the likelihood of single-peaked preferences varies significantly for the Impartial (Anonymous) Culture assumption, the P\'{o}lya urn and the Mallows model.
For elections chosen according to the IC or the IAC assumption, it is extremely unlikely that single-peakedness arises (cf. Table~\ref{tab:results-ic}).
With Theorem~\ref{thm:SW-bound}, we have shown that unlikeliness also holds for arbitrary domain restrictions that avoid a $(2,k)$-configuration.
In contrast, for the P\'{o}lya urn and the Mallows model with parameter $a$ ($\phi$) chosen sufficiently large (small) it is rather likely that elections are single-peaked.
Numerical probabilities in Table~\ref{tab:results-polya} and \ref{tab:results-mallows} affirm this claim.

Let us conclude with directions for future research.
Theorem~\ref{thm:SW-bound} requires that the domain restriction avoids a $(2,k)$-configuration and thus is not applicable to domain restrictions such as the single-crossing restriction \citep{roberts1977voting,DBLP:journals/scw/BredereckCW13} or 2D single-peaked restriction \citep{barbera1993generalized}.
It remains open whether this result can be extended to such domain restrictions as well and how the corresponding bound would look like.
It would also be interesting to complement Theorem~\ref{thm:SW-bound} with a corresponding lower bound result.
In general, the likelihood of other domain restrictions such as the single-crossing \citep{roberts1977voting} or the 2D single-peaked restriction \citep{barbera1993generalized} has yet to be studied.

In Section~\ref{sec:iac} we studied the likelihood of single-peakedness under the IAC assumption.
In particular, it follows from Theorem~\ref{thm:iac} that under IAC the likelihood that an $(n,3)$-election is single-peaked is $\frac{60n}{(n+2)(n+3)(n+4)}$. The likelihood that an $(n,3)$-election has a Condorcet winner is $\frac{15(n+3)^2}{16(n+2)(n+4)}$ for odd $n$ and $\frac{15(n+2)(n^2+8n+8)}{16(n+1)(n+3)(n+5)}$ for even $n$ \citep{gehrlein02}. Note that the probability for single-peakedness is significantly smaller than the latter two and, in particular, the former converges to 0 whereas the latter converge to $15/16$ for $n\rightarrow \infty$. Recently, the top monotonicity restriction has been proposed \citep{barbera_top_2011} which is a generalization of the single-peaked and single-crossing domain but still guarantees a Condorcet winner.
It would be highly interesting to know the likelihood of top monotonicity restricted preferences and whether this probability is non-zero for $n\rightarrow \infty$.

Another direction is to consider other probability distributions such as the Plackett-Luce model~\citep{plackett,luce2005individual} or Mallows mixture models where more than one reference vote is considered~\citep{Mixture}.
One could also analyze the probability distribution that arises when assuming that all elections are single-peaked and that all elections of the same size are equally likely.
This would allow  allow to ask questions such as ``How likely is it that a single-peaked election is also single-crossing?''.
Finally, a recent research direction is to consider elections that are nearly single-peaked, i.e., elections that have a small distance to being single-peaked according to some notion of distance \citep{fal-hem-hem:c:nearly-single-peaked,elk-fal-sli:c:clone-structures,cor-gal-spa:c:single-peaked-width,ijcai/CornazGS13,aspectsofSP,bredereck-nearly}.
The likelihood that elections are nearly single-peaked remains a worthwhile direction for future research.

\section*{Acknowledgements}
The first author was supported by the Austrian Science Foundation FWF, grant P25337-N23, the second author by the FWF, grant P25518-N23 and Y698 and by the European Research Council (ERC) under grant number 639945 (ACCORD).

We would like to thank Jiehua Chen for pointing out mistakes in the proof of Theorem~\ref{thm:SP_small_profiles}~(iii); the  result itself, as published~\citep{scw/LacknerL-likelihoodSP}, remains correct.

\bibliographystyle{abbrvnat}
\bibliography{../lit,../single_peak,../lit-new}

\end{document}